\documentclass[11pt]{article}
\usepackage{geometry}
\input{preamble}

\begin{document}

%
%

\title{Near-Linear Time Approximation Schemes for Clustering in Doubling 
Metrics} 
\author[1]{Vincent Cohen-Addad}
\author[2]{Andreas Emil Feldmann\thanks{Supported by the Czech Science
Foundation GA{\v C}R (grant \#19-27871X), and by the Center for Foundations of
Modern Computer Science (Charles Univ.\ project UNCE/SCI/004).}}
\author[3]{David Saulpic}
\affil[1]{Sorbonne Universit\'e, UPMC Univ Paris 06, CNRS, LIP6, Paris, France (\texttt{vcohenad@gmail.com})}
\affil[2]{Charles University in Prague, Czechia 
(\texttt{feldmann.a.e@gmail.com})}
\affil[3]{\'Ecole Normale Sup\'erieure, Paris (\texttt{david.saulpic@ens.fr})}

\date{}
\maketitle 
 
\begin{abstract}
  We consider the classic Facility Location, $k$-Median, and $k$-Means problems 
in metric spaces of doubling dimension $d$. We give nearly linear-time 
approximation schemes for each problem. The complexity of our algorithms is 
$\widetilde O(2^{(1/\eps)^{O(d^2)}} n)$, making a significant improvement over 
the state-of-the-art algorithms which run in time $n^{(d/\eps)^{O(d)}}$.

  Moreover, we show how to extend the techniques used to get the first efficient 
approximation schemes for the problems of prize-collecting $k$-Median and 
$k$-Means, and efficient bicriteria approximation schemes for $k$-Median with 
outliers, $k$-Means with outliers and $k$-Center.
\end{abstract}

\maketitle 

%
\newpage

\section{Introduction}

The $k$-Median and $k$-Means problems are classic clustering problems that are 
highly popular for modeling the problem of computing a ``good'' partition of a 
set of points of a metric space into $k$ parts so that points that are ``close'' 
should be in the same part. Since a good clustering of a dataset allows to 
retrieve information from the underlying data, the $k$-Median and $k$-Means 
problems are the cornerstone of various approaches in data analysis and machine 
learning. The design of efficient algorithms for these clustering problems has 
thus become an important challenge.

The input for the problems is a set of points in a metric space
and the objective is to identify a set of $k$ centers $C$ such that the
sum of the $p$th power of the distance from each point of the metric to its 
closest center in $C$ is minimized. In the $k$-Median problem, $p$ is set to 1
while in the $k$-Means problem, $p$ is set to 2.
In general metric spaces both problems are known to be
APX-hard, and this hardness even extends to Euclidean
spaces of any dimension $d = \Omega(\log n)$~\cite{AwasthiCKS15}.
Both problems also remain NP-hard for points in
$\R^{2}$~\cite{MegiddoS84}. For $k$-Center, the goal is to minimize the maximum distance from each
point in the metric to its closest center. This problem is APX-hard even in Euclidean
Spaces~\cite{feder1988optimal}, and computing a solution with optimal cost but $(1+\eps)k$ centers 
requires time at least~$\Omega(n^{\sqrt{1/\eps}})$~\cite{marx2015optimal}. 
Therefore, to get an efficient approximation scheme one needs to approximate 
both 
the number of centers and the cost. (See Section~\ref{sec:relatedwork} for more related 
work). 

To bypass these hardness of approximation results, researchers have considered
low-dimensional inputs like Euclidean spaces of fixed
dimension or more generally metrics of fixed doubling dimension.
There has been a large body of work to design good
tools for clustering in metrics of fixed doubling dimension,
from the general result of Talwar~\cite{talwar2004bypassing}
to very recent 
coreset
constructions for clustering problems~\cite{huang2018varepsilon}.
In their seminal work, Arora et al.~\cite{AroraRagRao98} gave a polynomial time 
approximation scheme (PTAS) for $k$-Median
in~$\R^2$, which generalizes to a quasi-polynomial time approximation scheme 
(QPTAS) for inputs in $\R^d$.
This result was improved in two ways.
First by Talwar~\cite{talwar2004bypassing} who generalized the result to any 
metric space
of fixed doubling dimension. Second by Kolliopoulos and 
Rao~\cite{kolliopoulos2007nearly}
who obtained an $f(\eps,d)\cdot n \log^{d+6} n$ time algorithm for $k$-Median
in $d$-dimensional Euclidean space. Unfortunately, Kolliopoulos
and Rao's algorithm relies on the Euclidean structure of the
input and does not immediately generalize to low dimensional doubling metric.
Thus, until recently the only result known for $k$-Median
in metrics of fixed doubling dimension was a QPTAS. This was also
the case for slightly simpler problems such as Uniform Facility Location.
Moreover, as pointed out in \cite{Cohen-Addad:2018},
the classic approach of Arora et al.~\cite{AroraRagRao98}
cannot work for the $k$-Means problem. Thus no efficient algorithms were known
for the $k$-Means problem, even in the plane. 

Recently, Friggstad et al.~\cite{friggstad2016local} and Cohen-Addad et 
al.~\cite{Cohen-AddadKM16}
showed that the classic local search algorithm for the problems gives
a $(1+\eps)$-approximation in time $n^{1/\eps^{O(d)}}$ in Euclidean
space, in time $n^{O(1/\eps^2)}$ for planar graphs (which also extends
to minor-free graphs), and
in time $n^{(d/\eps)^{O(d)}}$ in metrics of
doubling dimension $d$~\cite{friggstad2016local}. More recently 
Cohen-Addad~\cite{Cohen-Addad:2018}
showed how to speed up the local search algorithm for Euclidean space
to obtain a PTAS with running time $n k (\log n)^{(d/\eps)^{O(d)}}$.

Nonetheless, obtaining an efficient approximation scheme
(namely an algorithm running in time $f(\eps,d) \text{poly}(n)$)
for $k$-Median and $k$-Means in metrics of doubling dimension $d$
has remained a major challenge.

\bigskip
The versatility of the techniques we develop to tackle these problems allows us 
to consider a broader setting, where the clients do not necessarily have to be 
served. In the prize-collecting version of the problems, every client has a 
penalty cost that can be paid instead of its serving cost. In the $k$-Median 
(resp.~$k$-Means) with outliers problems, the goal is to serve all but $z$ 
clients, and the cost is measured on the remaining ones with the $k$-Median 
(resp.~$k$-Means) cost. These objectives can help to handle some noise from the 
input: the $k$-Median objective can be dramatically perturbed by the addition of 
a few distant clients, which must then be discarded. 

\subsection{Our Results}

We solve this open problem by proposing the first near-linear time algorithms 
for the $k$-Median and $k$-Means problems in metrics of fixed doubling 
dimension. More precisely, we show the following theorems.


\begin{theorem}
  \label{thm:clusteringapproxmedians}
  For any $0 <\eps < 1/3$, there exists a randomized $(1+\eps)$-approximation 
algorithm for $k$-Median in metrics of doubling dimension $d$ with running time 
$2^{(1/\eps)^{O(d^2)}} n \log^4 n + 2^{O(d)} n\log^9 n$ and success 
probability 
at least $1-2\eps$.
\end{theorem}

\begin{theorem}
  \label{thm:clusteringapproxmeans}
  For any $0 <\eps < 1/3$, 
  there exists a randomized $(1+\eps)$-approximation algorithm for $k$-Means in 
metrics of doubling dimension $d$ with running time $2^{(1/\eps)^{O(d^2)}} n 
\log^5 n + 2^{O(d)} n\log^{9} n$ and success probability at least $1-O(\eps)$.
\end{theorem}

Our results also extend to the Facility Location problem, in which no bound on 
the number of opened centers is given, but each center comes with an opening cost. 
The aim is to minimize the sum of the ($1$st power) of the distances from each point 
of the metric to its closest center, in addition to the total opening costs of all used centers.

\begin{theorem}\label{thm:fl}
   For any $0 < \eps <1/3$, there exists a randomized $(1+\eps)$-approximation 
algorithm for Facility Location in metrics of doubling dimension $d$ with 
running time $2^{(1/\eps)^{O(d^2)}}n + 2^{O(d)}n\log n$ and success probability 
at least $1-\eps$.
\end{theorem}

In all these theorems, we make the common assumption to have access to the distances
of the metric in constant time, as, e.g., in \cite{har2006, cole2006,Gottlieb15}.
This assumption is discussed in Bartal et al.~\cite{bartal2011fast}.

Note that the double-exponential dependence on $d$ is unavoidable unless P = NP, 
since the problems are APX-hard in Euclidean space of dimension $d = O(\log n)$. 
For Euclidean inputs, our algorithms for the $k$-Means and $k$-Median problems 
outperform the ones of Cohen-Addad~\cite{Cohen-Addad:2018}, removing in particular 
the dependence on $k$, and the one of Kolliopoulos and Rao~\cite{kolliopoulos2007nearly} 
when $d > 3$, by removing the dependence on $\log^{d+6} n$. Interestingly, for 
$k = \omega(\log^9 n)$ our algorithm for the $k$-Means problem is faster than 
popular heuristics like $k$-Means++ which runs in time $O(n k)$ in Euclidean space. 

We note that the success probability can be boosted to $1-\eps^{\delta}$
by repeating the algorithm $\log \delta$ times and outputting the
best solution encountered.

After proving the three theorems above, we will apply the techniques to prove the 
following ones.
We say an algorithm is an $(\alpha, \beta)$-approximation for $k$-Median or 
$k$-Means
with outliers if its cost is within an $\alpha$ 
factor of the optimal one and the solution drops $\beta z$ outliers.
Similarly, an algorithm is an $(\alpha, \beta)$-approximation for $k$-Center 
if its cost is within an $\alpha$ 
factor of the optimal one and the solution opens $\beta k$ centers.

\begin{theorem}\label{thm:pckmed}
  For any $0 <\eps < 1/3$, there exists a randomized $(1+\eps)$-approximation 
algorithm for Prize-Collecting $k$-Median (resp. $k$-Means) in metrics of doubling dimension $d$ 
with running time $2^{(1/\eps)^{O(d^2)}} n \log ^4 n + 2^{O(d)} n\log^9 n$ and 
success probability at least $1-\eps$.
\end{theorem}

\begin{theorem}\label{thm:kmedoutliers}
  For any $0 <\eps < 1/3$, there exists a randomized $(1+\eps, 1+O(\eps))$-approximation 
algorithm for $k$-Median (resp.~$k$-Means) with outliers in metrics of doubling 
dimension $d$ with running time $2^{(1/\eps)^{O(d^2)}} n \log ^6 n + T(n)$ and 
success probability at least $1-\eps$, where $T(n)$ is the running time to 
construct a constant-factor approximation.
\end{theorem}

We note as an aside that our proof of Theorem~\ref{thm:kmedoutliers} could give an approximation
where at most $z$ outliers are dropped, but $(1+O(\eps))k$ centers are opened. 
For simplicity, we focused on the previous case.

\begin{theorem}\label{thm:kcenter}
   For any $0 <\eps < 1/3$, there exists a randomized $(1+\eps, 1+O(\eps))$-approximation 
algorithm for $k$-Center in metrics of doubling 
dimension $d$, with running time $2^{(1/\eps)^{O(d^2)}} n \log ^6 n + n \log k$ 
and success probability at least $1-\eps$.
\end{theorem}

As explained above, this bicriteria is necessary in order to get an efficient algorithm:
it is APX-hard to approximate the cost~\cite{feder1988optimal}, and achieving
the optimal cost with $(1+\eps)k$ centers requires a complexity 
$\Omega(n^{1/\sqrt{\eps}})$~\cite{marx2015optimal}.  To the best of our knowledge, this works presents the first linear-time bicriteria approximation scheme for the problem of $k$-center.

\subsection{Techniques}
To give a detailed insight on our techniques and our contribution we first need to quickly
review previous approaches for obtaining approximation schemes on bounded doubling metrics.
The general approach, due to Arora~\cite{arora1998polynomial} and Mitchell~\cite{mitchell1999guillotine}, 
 which was generalized to doubling metrics by 
Talwar~\cite{talwar2004bypassing}, is the following.

\subsubsection{Previous Techniques}
The approach consists in randomly partitioning the metric into a constant number
of regions, and applying this recursively to each region. The recursion stops
when the regions contain only a constant number of
input points. This leads to what is called a \emph{split-tree decomposition}:
a partition of the space into a finer and finer level of granularity.
The reader who is not familiar with the split-tree decomposition may refer to Section~\ref{sec:decomp} for a more formal introduction.

\paragraph{Portals}
The approach then identifies a specific set of points for each region, called \emph{portals},
which allows to show that there exists a near-optimal solution such that different regions
``interplay'' only through portals. For example, in the case of the Traveling Salesperson (TSP) problem,
it is possible to show that there exists a near-optimal tour that enters and leaves a region only
through its portals.
In the case of the $k$-Median problem a client located in a specific
region can be assigned to a facility in a different region only through a path that goes to a portal
of the region. In other words, clients can ``leave'' a region only through the portals. 

Proving the existence of such a structured near-optimal solution relies on the fact that the probability
that two very close points end up in different regions of large diameter 
is very unlikely. Hence the expected \emph{detour} paid by going through a portal of the region
is small compared to the original distance between the two points, if the 
portals are dense enough.

For the sake of argument, we provide a proof sketch of the standard proof of 
Arora~\cite{arora1998polynomial}. We will use a refined version of this idea in
later sections. 
The split-tree recursively divides the input metric $(V,\dist)$ into parts of 
smaller and smaller diameter. The root 
part consists of the entire point set and the parts at level $i$ are of 
diameter roughly $2^{i}$. The set of portals of a part of level $i$ is an 
\emph{$\eps_0 2^i$-net} for some $\eps_0$, which is a small set such that every point of the metric is at distance at most~$\eps_0 2^i$ to it. Consider two points $u,v$ and let us 
bound the expected detour incurred by connecting $u$ to $v$ through portals. 
This detour is determined by a path that starting from $u$ at the lowest level, 
in each step connects a vertex $w$ to its closest net point of the part 
containing $w$ on the next higher level. This is done until the lowest-level 
part $R_{u,v}$ (i.e., the part of smallest diameter) is reached, which 
contains both $u$ and $v$, from where a similar procedure leads from this level 
through portals of smaller and smaller levels all the way down to $v$. If the 
level of $R_{u,v}$ is~$i$ then the detour, i.e., the difference between 
$\dist(u,v)$ and the length of the path connecting $u$ and $v$ through portals, 
is $O(\eps_0 2^i)$ by the definition of the net. Moreover, the proof shows that 
the probability that $u$ and $v$ are not in the same part on level $i$ is at 
most $\dist(u,v)/2^i$. Thus, the expected detour for connecting $u$ to $v$ is 
$\sum_{\text{level } i} \text{Pr}[R_{u,v} \text{ is at level } i] \cdot O(\eps_0 
2^i) = \sum_{\text{level } i} O(\eps_0 \dist(u,v))$. Hence, setting $\eps_0$ to 
be some~$\eps$ divided by the number of levels yields that the expected detour 
is $O(\eps \dist(u,v))$.

\paragraph{Dynamic programming}
The portals now act as separators between different parts and allows to apply a dynamic programming (DP)
approach for solving the problems.
The DP consists of a DP-table entry for each part and for each 
\emph{configuration}  of the portals of the part. Here a configuration is a potential way 
the near-optimal solution interacts with the part.
For example, in the case of TSP, a configuration is the information at which portal the near-optimal
tour enters and leaves and how it connects the portals on the outside and inside of the part.
For the $k$-Median problem, a configuration stores how many clients outside (respectively inside) the part connect through each portal and are served
by a center located inside (respectively outside).
Then the dynamic program proceeds in a bottom-up fashion along the split-tree to fill up the DP table.
The running time of the dynamic program depends exponentially on the number of portals.

\paragraph{How many portals?}
The challenges that need to be overcome when applying this approach, and in particular
to clustering problems, are two-fold.
First the ``standard'' use of the split-tree requires $O((\frac{\log n}{\eps})^d)$ portals per part
in order to obtain a $(1+\eps)$-approximation, coming from the fact that the 
number of levels can be assumed to be
logarithmic in the number of input points.
This often implies quasi-polynomial time approximation schemes since the running
time of the dynamic program has exponential dependence on the number of portals.
This is indeed the
case in the original paper by Talwar~\cite{talwar2004bypassing} and
in the first result on clustering in Euclidean space by Arora et al.~\cite{AroraRagRao98}.
However, in some cases, one can lower the number of portals per part needed. In Euclidean
space for example, the celebrated ``patching lemma''~\cite{arora1997nearly} shows 
that only a constant number (depending on $\eps$) of portals
are needed for TSP.
Similarly, Kolliopoulos and Rao~\cite{kolliopoulos2007nearly} showed that for $k$-Median
in Euclidean space only a constant number of portal are needed, if one uses a slightly different
decomposition of the metric.

Surprisingly, obtaining such a result for doubling metrics is much more challenging. 
To the best of our knowledge, this work is the first one to reduce the number of 
portals to a constant.

\medskip

A second challenge when working with split-tree decompositions and the $k$-Means problem is that because
the cost of assigning a point to a center is the squared distance, the analysis of Arora, Mitchell, and
Talwar does not apply. If two points are separated at a high level of the split-tree, then making a detour
to the closest portal may incur an expected cost much higher than the cost of the optimal solution. 

\subsubsection{Our Contributions}

Our contribution can be viewed as a ``patching lemma'' for clustering problems 
in doubling metrics. Namely, an approach that allows to solve the problems 
mentioned above: (1)~it shows how to reduce the number of portals to a constant, 
 (2)~it works for any 
clustering objective which is defined as the sum of distances to some constant 
$p$ (with $k$-Median and $k$-Means as prominent special cases), and (3)~it works 
not only for Euclidean but also for doubling metrics.

Our starting point is the notion of \emph{badly cut} vertices of 
Cohen-Addad~\cite{CA_capac} for the capacitated version of the above clustering 
problems. To provide some intuition on the definition, let us focus on $k$-median and start with the 
following observation: consider a center $f$ of the optimal solution and a 
client $c$ assigned to~$f$. If the diameter of the lowest-level part 
containing both $f$ and $c$ is of order $\dist(c,f)$ (say at most 
$\dist(c,f)/\eps^2$), then by taking a large enough but constant size net as a 
set of portals in each part (say an $\eps^3 2^i$-net for a part of level 
$i$), the total detour for the two points is at most $O(\eps \dist(c,f))$, which 
is acceptable.

The problematic scenario is when the lowest-level part containing $f$ and $c$ 
is of diameter much larger than $\dist(c,f)$. In this case, it is impossible to 
afford a detour proportional to the diameter of the part in the case of the 
$k$-Median and $k$-Means objectives.  To handle this case we first compute a 
constant approximation $L$ 
(via some known algorithm) and use it to guide us towards a 
$(1+\eps)$-approximation.

\begin{figure}[t]
\centering
\includegraphics[scale=0.3]{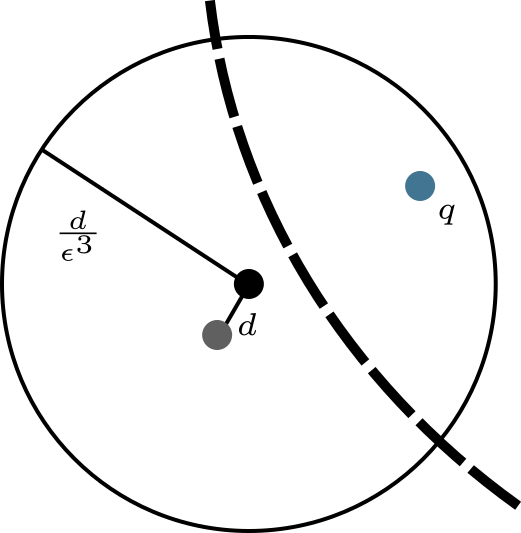}
 \caption{Illustration of badly cut. The black point is $c$ (resp.~$l$), the 
gray one is $L(c)$ (resp.~$f_0$), and the blue point is q. The dashed line is 
the boundary of a part with ``large'' diameter.}
 \label{fig:badcut}
\end{figure}

\paragraph{Badly cut clients and facilities}
Consider a client~$c$ and the center $L(c)$ serving $c$ in $L$ (i.e., $L(c)$ is 
closest to $c$ among the centers in $L$), and call $\opt(c)$ the facility of an 
optimum solution $\opt$ that serves $c$ in $\opt$. We say that $c$ is 
\emph{badly cut} if there is a point $q$ in the ball centered at $c$ of radius 
$\dist(c, L(c))/\eps$ such that the highest-level part containing $c$ and not 
$q$ is of diameter much larger than $\dist(c,L(c))/\eps$ (say greater than 
$\dist(c,L(c))/\eps^2$). In other words, there is a point $q$ in this ball such 
that paying a detour through the portal to connect $c$ to $q$ yields a detour 
larger than $\eps \dist(c,q)$ (see Figure~\ref{fig:badcut}).

Similarly, we say that a center $l$ is \emph{badly cut} if there is a point $q$ in the 
ball centered at $l$ of radius $\dist(l, f_0)/\eps$ (where $f_0$ is the 
facility of $\opt$ that is the closest to $l$) such that the highest-level 
part containing $l$ and not  $q$ is of diameter $\dist(l,f_0)/\eps^2$. 
The crucial property here is that any client 
$c$ or any facility $l$ is badly cut with probability $O(\eps^3)$, as we will 
show.

\paragraph{Using the notion of badly cut}
We now illustrate how this notion can help us. Assume for simplicity that 
$\opt(c)$ is in the ball centered at a client $c$ of radius $\dist(c, 
L(c))/\eps$ (if this is not the case then $\dist(c, \opt(c))$ is much larger 
than $\dist(c, L(c))$, so this is a less problematic scenario and a simple idea 
can handle it). If $c$ is not badly cut, then the lowest-level part containing 
both $c$ and $\opt(c)$ is of diameter not much larger than $\dist(c, 
L(c))/\eps$. Taking a sufficiently fine net for each part (independent of 
the number of levels) allows to bound the detour through the portals to reach 
$\opt(c)$ from $c$ by at most $\eps \dist(c, L(c))$. Since $L$ is an 
$O(1)$-approximation, this is fine.

If $c$ is badly cut, then we modify the instance by relocating $c$
to $L(c)$. That is, we will work with the instance where there is no more client
at $c$ and there is an additional client at $L(c)$. We claim that any solution
in the modified instance can be lifted to the original instance at an
expected additional cost of $O(\eps^3 \opt)$. This comes from the fact
that the cost increase
for a solution is, by the triangle inequality, at most the sum of distances
of the badly cut clients to their closest facility in the local solution.
This is at most $O(\eps^3 \opt)$ in expectation 
since each client is badly cut with probability at most $O(\eps^3)$ and
$L$ is an $O(1)$-approximation.

Here we should ask, what did we achieve by moving $c$ to $L(c)$? Note that
$c$ should now be assigned to facility $f$ of $\opt$ that is the closest to $L(c)$.
So we can make the following observation:
If $L(c)$ is not badly cut, then the detour through the portals when assigning
$c$ to $f$ is fine (namely at most $\eps$ times the distance from $L(c)$ to
its closest facility in $\opt$). Otherwise, if $L(c)$ is also badly cut,
then we simply argue that there exists a near-optimal
solution which contains $L(c)$,
in which case $c$ is now served optimally at a cost of 0
(in the new instance).

\paragraph{From bicriteria to opening exactly $k$ centers} 
Since $L(c)$ is badly cut with probability $O(\eps^3)$, this leads to a solution 
opening $(1+O(\eps^3))k$ centers. At first, it looks difficult to then reduce 
the number of centers to $k$ without increasing the cost of the solution by a 
factor larger than $(1+\eps)$. However, and perhaps surprisingly, we show in 
Lemma~\ref{lem:port-resp-medians} that this can be avoided: we show that there exists 
a near-optimal solution that contains the badly cut centers of $L(c)$. 

We can then conclude that a near-optimal solution can be computed by a simple 
dynamic-programming procedure on the split-tree decomposition to identify the 
best solution in the modified instance.

Our result on Facility Location in Section~\ref{sec:FL} provides a simple illustration 
of these ideas --- avoiding the bicriteria issue due to the hard bound on the 
number of opened facilities for the $k$-Median and $k$-Means problems. Our main 
result on $k$-Median and $k$-Means is described in Section~\ref{sec:structkmedian}. We discuss
some extensions of the framework in Section~\ref{sec:extension}.

\subsection{Related work} 
\label{sec:relatedwork}

\paragraph{On clustering problems}
The clustering problems considered in this paper are known to be NP-hard, 
even restricted to inputs
lying in the Euclidean plane (see Mahajan et al.~\cite{MahajanNV12} or Dasgupta 
and Freund~\cite{dasgupta2009random} for $k$-Means, Megiddo and 
Supowit~\cite{megiddo1984complexity} for the problems with outliers, and Masuyama 
et al.~\cite{masuyama1981computational} for $k$-Center). The problems of 
Facility Location and $k$-Median have been studied since a long time
in graphs, see e.g.~\cite{kanungo2004local}. The current best approximation 
ratio for metric Facility
Location is 1.488, due to Li~\cite{Li13}, whereas it is 2.67 for $k$-Median, due to
Byrka et al.~\cite{ByrkaPRST15}.

The problem of $k$-Means in general graphs also received a lot of attention (see
e.g., Kanungo et al.~\cite{kanungo2004local}) and the best approximation ratio
is 6.357, due to Ahmadian et al.~\cite{ahmadian2017better}.

Clustering problems with outliers where first studied by Charikar et 
al.~\cite{CharikarKMN01}, who devised an $(O(1), (1+O(\eps))$-approximation for 
$k$-Median with outliers and a constant factor approximation for prize-collecting $k$-Median.
More recently, Friggstad et al.~\cite{friggstad2019approximation} showed that 
local search provides a bicriteria
approximation, where the number of centers is approximate instead of the number
of outliers. However, the runtime is $n^{f(\eps, d)}$, and thus we provide a 
much faster algorithm. To the best of our knowledge, we present the first 
approximation scheme that preserves the number of centers.

The $k$-Center problem is known to be NP-hard to approximate within any 
factor better than~2, a
bound that can be achieved by a greedy algorithm~\cite{feder1988optimal}. This
is related to the problem of covering points with a minimum number of disks (see
e.g.~\cite{liao2010polynomial,marx2015optimal}). 
Marx and Pilipczuk~\cite{marx2015optimal} proposed an exact algorithm running 
in time $n^{\sqrt{k} + O(1)}$ to find the maximum number of points covered by $k$
disks and showed a matching lower bound, whereas Liao et al.~\cite{liao2010polynomial}
presented an algorithm running in time $O(mn^{O(1/\eps^2 \log^2 1/\eps)})$ to
find a $(1+\eps)$-approximation to the minimal number of disks necessary to cover
all the points (where $m$ is the total number of disks and $n$ the number of 
points).
This problem is closely related to $k$-Center: the optimal value of $k$-Center
on a set $V$ is the minimal number $L$ such that there exist $k$ disks of 
radius $L$ centered on 
points of $V$ covering all points of~$V$. Hence, the algorithm 
from~\cite{liao2010polynomial}
can be directly extended to find a solution to $k$-Center with $(1+\eps)k$ centers 
and optimal cost. The local search algorithm of Cohen-Addad et al.~\cite{} can be adapted to $k$-center and generalizes the last result to any dimension $d$: in $\R^d$, one can find a solution with optimal cost and $(1+\eps)k$ centers in time $n^{1/\epsilon^{O(d)}}$. Loosing on the approximation allows us to present a much faster algorithm.

\paragraph{On doubling dimension}
Despite their hardness in general metrics, these problems admit a PTAS when the
input is restricted to a low dimensional metric space: Friggstad et 
al.~\cite{friggstad2016local}
showed that local search gives a $(1+\eps)$-approximation. However, the running
time of their algorithm is $n^{(d/\eps)^{O(d)}}$ in metrics 
with doubling dimension $d$.

A long line of research exists on filling the gap between results for Euclidean 
spaces
and metrics with bounded doubling dimension. This started with the 
work of Talwar~\cite{talwar2004bypassing}, who gave QPTASs for a long list of 
problems.
The complexity for some of these problems was improved later on: for the Traveling 
Salesperson problem, Gottlieb~\cite{Gottlieb15} gave a near-linear time 
approximation scheme,
Chan et al.~\cite{chan2016ptas} gave a PTAS for Steiner Forest, and 
Gottlieb~\cite{Gottlieb15} described an efficient spanner construction.


\section{Preliminaries}
\label{sec:prelims}

\subsection{Definitions}
\label{sec:problemdef}
Consider a metric space $(V,\dist)$.
For a vertex $v \in V$ and an integer $r \ge 0$, we let
$\beta(v,r) = \{w\in V \mid \dist(v,w) \le r \}$ be the
\emph{ball} around $v$ with radius $r$. 
The \emph{doubling dimension} of a metric is the smallest integer $d$ such that any
ball of radius $2r$ can be covered by $2^d$ balls of radius~$r$. We call $\Delta$
the aspect-ratio (sometimes referred to as \emph{spread} in the literature) of the metric, i.e., the ratio between the largest and the smallest distance.

Given a set of points called \emph{clients} and a set of
points called \emph{candidate centers} in
a metric space, the goal of the \emph{$k$-Median} problem
is to output a set of $k$ \emph{centers} (or \emph{facilities}) chosen among 
the candidate centers
that minimizes the sum  of the distances from each client to its closest center.
More formally, an instance to the 
$k$-Median problem is a 4-tuple $(C, F, \dist, k)$, where $(C \cup F, \dist)$ 
is a metric space and $k$ is a positive integer. The goal is to find a set $S 
\subseteq F$ such that $|S| \le k$ and $\sum_{c \in C} \min_{f \in 
  S}(\dist(c,f))$ is minimized. Let $n = |C \cup F|$.
The \emph{$k$-Means} problem is identical except from the objective function 
which is $\sum_{c \in C} \min_{f \in S}(\dist(c,f))^2$.

In the \emph{Facility Location} problem, the number of centers in the solution
is not limited but there is a cost $w_f$ for each candidate center $f$ and
the goal is to find a solution $S$ minimizing 
$\sum_{c \in C} \min_{f \in S}(\dist(c,f)) + \sum_{f \in S} w_f$.

For those clustering problems, it is convenient to name the center serving a client. For a client $c$ and a solution $S$, we denote $S(c)$ the center closest to $c$, and $S_c := \dist(c, S(c))$ the distance to it.
 
In this paper, we consider the case where the set of candidate centers
is part of the input. A variant of the $k$-Median and $k$-Means problems
in Euclidean metrics allows to place centers anywhere in the space and specifies
the input size as simply the number of clients. We note that up to losing a 
polylogarithmic factor in the running time, it is possible to reduce 
this variant to our setting by computing a set of candidate centers that 
approximate the best set of centers in~$\R^d$~\cite{matouvsek2000approximate}.

A $\delta$-\emph{net} of $V$ is a set of points $X\subseteq V$ such that for all
$v \in V$ there is an $x \in X$ such that $\dist(v, x) \leq \delta$, and 
for all $x, y \in X$ we have $\dist(x, y) > \delta$. A net is therefore a set 
of points not too close to each other, such that
every point of the metric is close to a net point. 
The following lemma bounds the cardinality of a net in doubling metrics.

\begin{lemma}[from Gupta et. al \cite{gupta2003bounded}]\label{prop:doub:net}
Let $(V, d)$ by a metric space with doubling dimension $d$ and diameter 
$D$, and let $X$ be a $\delta$-net of $V$. Then $|X| \leq 2^{d \cdot 
\lceil \log_2 (D/\delta)\rceil}$.
\end{lemma}

Another property of doubling metrics that will be useful for our purpose is
the existence of low-stretch spanners with a linear number of edges. 
More precisely, Har-Peled and Mendel~\cite{har2006} showed that one can find a graph 
(called a \emph{spanner})
in the input metric that has $O(n)$ edges
 such that distances in the graph approximate the original 
distances up to a constant factor. This construction takes time $2^{O(d)} n$.
We will make use of these spanners only for computing constant-factor approximations of 
our problems: for this purpose, we will therefore assume that the number of
edges is $m = 2^{O(d)}n$.

We will also make use the following lemma.

\begin{lemma}[\cite{Cohen-AddadS17}]
  \label{faketriangleineq}
  Let $p \ge 0$ and $1/2 > x>0$. For any $a,b,c \in V$, we have
  $\dist(a,b)^p \le (1+x)^p \dist(a,c)^p + \dist(c,b)^p(1+1/x)^p$.
\end{lemma}

\subsection{Decomposition of Metric Spaces}\label{sec:decomp}

As pointed out in our techniques section, we will make use of hierarchical
decompositions of the input metric.
We define a \emph{hierarchical decomposition}
(sometimes simply a decomposition) of a metric $(V, \dist)$
 as a collection of partitions $\calD = 
\{\calB_0,\ldots,\calB_{|\calD|}\}$ that satisfies the following:
\begin{itemize}
\item each $\calB_i$ is a partition of $V$,
\item $\calB_i$ is a refinement of $\calB_{i+1}$,
  namely for each part 
$B\in\calB_i$ there exists a part $B'\in\calB_{i+1}$ that contains~$B$,
\item $\calB_0$ contains a singleton set for each $v\in V$, while 
  $\calB_{|\calD|}$ is a trivial partition that contains only one
  set, namely $V$.
\end{itemize}

We define the \emph{$i$th level} of the decomposition to be the
partition $\calB_{i}$, and call $B \in \calB_i$ a level-$i$ part. If $B' \in 
\calB_{i-1}$ is such that $B' \subset B$, we say that $B'$ is a 
\emph{subpart} of $B$.

For a given decomposition $\calD=\{\calB_0,\ldots,\calB_{|\calD|}\}$, we say 
that a vertex $u$ is \emph{cut from $v$ at level $j$} if $j$ is the maximum 
integer such that $v$ is in some $B \in \calB_j$ and $u$ is in some $B' \in 
\calB_j$ with $B \ne B'$. For a vertex $v$ and radius $r$ we say that the ball 
$\beta(v, r)$ is \emph{cut} by $\calD$ at level~$j$ if $j$ is the maximum level 
for which some vertex of the ball is cut from $v$ at level $j$.


A key ingredient for our result is the following lemma, that introduces 
some properties of the hierarchical decomposition (sometimes referred to as \emph{split-tree})
proposed by Talwar~\cite{talwar2004bypassing} for low-doubling metrics.

\newpage
\begin{lemma}[Reformulation of \cite{talwar2004bypassing, BartalG13}]\label{lem:talwar-decomp}
For any metric $(V,\dist)$ of doubling dimension $d$ and any $\rho > 0$, 
there is a randomized hierarchical 
decomposition $\calD$ such that the diameter of a part $B \in \calB_i$ is at most
$\growthrate^{i+1}$, $|\calD|\leq\lceil\log_\growthrate(\diam(V))\rceil$, each part $B \in \calB_i$ is refined in at most $2^{O(d)}$ parts at level $i-1$, and:

\begin{enumerate}

  \item\label{prop:doub:prob} \textbf{Scaling probability}: for any $v \in 
V$, radius $r$, and level $i$, we have 
    \[\Pr[\calD \text{ cuts } \beta(v,r) \text{ at a level } i] \leq  2^{2d+2} 
r/\growthrate^i.\]

 \item\label{prop:doub:portals}\textbf{Portal set:}
    every set $B\in \calB_i$ where $\calB_i\in\calD$ comes with a set of
    \emph{portals}~$\calP_B\subseteq B$ that is
    \begin{enumerate}
    \item \textbf{concise:} the size of the portal set is bounded by 
$|\calP_{B}| \le 1/\rho^d$;
      and
    \item \textbf{precise:} 
    for every node $u\in B$ there is a portal $p\in\calP_B$ close-by, i.e., 
$\dist(u,p)\leq \rho 2^{i+1}$; and
    \item \textbf{nested:} any portal of level $i+1$ that lies in $B$ is also a 
portal of $B$, i.e., for every $p\in\calP_{B'}\cap B$ where $B'\in\calB_{i+1}$ 
we have $p\in\calP_B$.
    \end{enumerate}
  \end{enumerate}
Moreover, this decomposition can be found in time $(1/\rho)^{O(d)} n \log \Delta$.
\end{lemma}

\subsection{Formal Definition of Badly Cut Vertices}

As sketched in the introduction, the notion of badly cut lies at the heart of our
analysis. We define it formally here.
We denote $\badcut = \eps^2\frac{p}{(p+\eps)^p}$ and 
$\offset = 2d + 2 + \log(1/\badcut)$,
two 
parameters that are often used throughout this paper.

\begin{definition}\label{def:badlycut}
  Let $(C\cup F,\dist)$ be a metric with doubling dimension $d$, let $\calD$ be 
a hierarchical decomposition of the metric, and $\eps > 0$. Let also $L$ be a 
solution to the instance for any of the problems Facility Location, $k$-Median, 
or $k$-Means. A client $v\in C$ is \emph{badly cut w.r.t. $\calD$} if the ball $\beta(v, 3L_v / \eps)$ is cut as some level $j$ greater than $\log(3L_v/ \eps) + \offset$, where $L_v$ is the distance from $v$ to the closest facility of $L$.


 Similarly, a center $f\in F$ of  $L$ is \emph{badly cut w.r.t.~$\calD$} 
 if $\beta(f, 3\opt_f)$ is cut at some level~$j$ greater than $\log (3\opt_f) + 
\offset$, where $\opt_{f}$ is the distance from $f$ to the closest facility of 
$\opt$.
\end{definition}

In the following, when $\calD$ is clear from the context we simply say
badly cut.
The following lemma bounds the probability of being badly cut. 

\begin{lemma}\label{lem:badlycutddim}
  Let $(C \cup F, \dist)$ be a metric, and $\calD$ a random hierarchical 
  decomposition given by Lemma~\ref{lem:talwar-decomp}.
  Let $v$ be a vertex in $C \cup F$. The probability that $v$ is badly cut 
  is at most $\badcut$.
\end{lemma}
\begin{proof}
  Consider first a vertex $v \in C$. By Property~\ref{prop:doub:prob}, the 
probability that a ball $\beta(v, r)$ is cut at level at least $j$ is 
at most $2^{2d+2} r /\growthrate^j$. Hence the probability that a 
ball $\beta(v,3L_v/ \eps)$
is cut at a level $j$ greater than $\log (3L_v/\eps) + 2+2d + \log (1/\badcut)$ 
is at most $\badcut$.

The proof for $v \in F$ is identical.
\end{proof}

\subsection{Preprocessing}\label{sec:preprocess}

In the following, we will work with the slightly more general
version of the
clustering problems where there is some \emph{demand} on each vertex:
there is a function $\chi : C \mapsto \{1,\ldots,n\}$ and the goal is to minimize 
$\sum_{c \in C} \chi(c) \cdot \min_{f \in S} \dist(c,f) + \sum_{f \in S} w_f$ for
the Facility Location problem, or
$\sum_{c \in C} \chi(c) \cdot \min_{f \in S} \dist(c,f)$ and
$\sum_{c \in C} \chi(c) \cdot \min_{f \in S} \dist(c,f)^2$ for $k$-Median and 
$k$-Means respectively. This also extends to any
$\sum_{c \in C} \chi(c) \cdot \min_{f \in S} \dist(c,f)^p$ with constant $p$.
For simplicity, we will consider in the proof that the client set is actually a multiset, where a client $c$ appears $\chi(c)$ times.

We will preprocess the input instance to transform it into several instances of the more
general clustering problem, ensuring that the aspect-ratio $\Delta$ of each instance is polynomial.
We defer this construction to Appendix~\ref{ap:prelims}.

\section{A Near-Linear Time Approximation Scheme for Non-Uniform Facility 
Location}
\label{sec:FL}\label{sec:struct}

To demonstrate the utility of the notion of badly cut, we show how to use it to 
get a near-linear time approximation scheme for Facility Location in metrics of 
bounded doubling dimension. In this context we refer to centers in the set $F$ 
of the input as facilities.

We first show a structural lemma that allows to focus on instances that do not
contain any badly cut client. Then, we prove that these instances have 
\emph{portal-respecting} solutions that are nearly optimal, and that can be 
computed with a dynamic program.
We conclude by providing a fast dynamic program, that takes advantage of all the 
structure provided before.

\subsection{An instance with small distortion}

Consider a metric space $(V, \dist)$ and an instance $\I$ of 
the Facility Location problem on $(V,\dist)$. Here we generalize slightly, as explained in Section~\ref{sec:preprocess}, and restrict the set of candidate center to a subset $C$ of $V$. Our 
first step is to show that, 
given $\I$, a randomized decomposition $\calD$ of $(V,\dist)$ and any solution 
$L$ for $\I$ on $(V,\dist)$, we can build an instance $\ID$ on the same metric 
(but different client and center sets) such that any solution $S$ has a similar 
cost in $\I$ and in $\ID$, and more importantly $\ID$ does not contain any badly 
cut client with respect to $\calD$. The definition of $\ID$ depends on the 
randomness of $\calD$.

Let $\cost_{\I_0}(S)=\sum_{c\in C} \min_{f\in S}(\dist(c,f))^p$ be the cost 
incurred by only the distances to the facilities in a solution $S$ to an 
instance $\I_0$, and let $\eps>0$. For any instance $\ID$ on $(V,\dist$), we 
let 
\[
\nu_{\ID} = \max_{\text{solution }S} \big\{\cost_\I(S) -  (1+2\eps) 
\cost_{\ID} (S), (1-2\eps) \cost_{\ID} (S) - \cost_\I(S)\big\}.
\]
Note that in the particular case of Facility Location, $p=1$, but we allow it to be more general in order to make the proof adjustable to $k$-Means. 
If $B_{\calD}$ denotes the set of badly cut facilities (w.r.t $\calD$) of the 
solution $L$ from which instance $\ID$ is constructed, we say that $\ID$ has 
\emph{small distortion w.r.t.~$\I$} if $\sum_{f \in B_{\calD}} w_f \le \eps 
\cdot \sum_{f \in L} w_f$, and $\nu_{\ID} \leq \eps \cost_\I(L)$, 
and there exists a solution $S$ that contains $B_{\calD}$ with
\begin{equation}\label{eqn:solution}
\cost_{\ID}(S) \le (1+O(\eps))\cost_{\I}(\opt) + O(\eps) 
\cost_{\I}(L). 
\end{equation}
When $\I$ is clear from the context we simply say that $\ID$ has small 
distortion.
We will show that the solution $\opt' = \opt \cup B_{\calD}$ (where $\opt$ is 
the optimal solution for the instance~$\I$) fulfills the condition 
of~\eqref{eqn:solution}.

\medskip
In the following, we will always work with a particular $\ID$ constructed from 
$\I$ and a precomputed approximate solution $L$ for $\I$ as follows:
$\I$ is transformed such that every badly cut client $c$ is moved to~$L(c)$, 
namely, $\chi(L(c))$ is increased by $\chi(c)$ after which we set $\chi(c) = 
0$. Recall however that we treat the client set as a multiset, so that 
$\cost_{\I_0}(S)$ counts the distance from $c$ to the closest facility 
$\chi(c)$ times.


What we would like to prove is that the optimal solution in $\I$ can be transformed
to a solution in $\ID$ with a small additional cost, and vice versa. 
The intuition behind this is the following: a client of the solution $L$
is badly cut with probability $\badcut$ (from Lemma~\ref{lem:badlycutddim}), 
hence every client contributes with $\badcut L_c$ to transform any solution $S$ 
for the instance~$\I$ to a solution for the instance $\ID$, and vice versa. 

However, we will need to convert a particular solution in $\ID$ (think of it as $\opt_{\ID}$)
to a solution in $\I$: this particular solution depends in the randomness of $\calD$, and this
short argument does not apply because of dependency issues. It is nevertheless
possible to prove that $\ID$ has a small distortion, as done in the following lemma.

\begin{lemma}\label{lem:lift}
  Given an instance $\I$ of Facility Location, a randomized decomposition $\calD$, an $\eps$ such that $0 < \eps < 1/4$ and a solution $L$, let $\ID$ be the instance obtained from $\I$ by moving every badly cut client $c$ to $L(c)$ (as described above).
  The probability that $\ID$ has small distortion is at least $1-\eps$, where 
the solution fulfilling~\eqref{eqn:solution} is $\opt'=\opt\cup B_\calD$.
\end{lemma}
\begin{proof}
  To show the lemma, we will show that $\E{\sum_{f \in B_{\calD}} w_f} \le 
\eps^2 \sum_{f \in L} w_f/2$ and $\E{\nu_{\ID}} \leq \eps^2 \cost(L)/2$. Then, Markov's 
inequality and a union bound over the probabilities of failure yield $\sum_{f 
\in B_{\calD}} w_f \le \eps\cdot \sum_{f \in L} w_f$ and $\nu_{\ID} \leq \eps 
\cost_\I(L)$. 
Since $\opt\subseteq\opt'$ and $(1-2\eps)\cost_{\ID}(\opt')-\cost_\I(\opt')\leq 
\nu_{\ID}\leq \eps\cost_\I(L)$, the cost 
incurred by connecting clients to facilities in $\opt'$ is 
\begin{align*}
\cost_{\ID}(\opt')&\leq 
(1+\frac{2\eps}{1-2\eps})(\cost_\I(\opt')+\eps\cost_\I(L))
&\text{(as $\nu_{\ID} \leq \eps \cost_\I(L)$)}
\\
&< (1+4\eps)(\cost_\I(\opt')+\eps\cost_\I(L))
&\text{(as $\eps< 1/4$)}
\\
&\leq (1+4\eps)(\cost_\I(\opt)+\eps\cost_\I(L))
&\text{(as $\opt\subseteq\opt'$)},
\end{align*}
which shows that $\ID$ has small distortion.

Note that $\E{\sum_{f \in B_{\calD}} w_f} = \sum_{f \in L} \Pr[f \text{ badly 
cut}] \cdot w_f \le \eps^2 \sum_{f \in L} w_f/2$ is immediate from 
Lemma~\ref{lem:badlycutddim}. It remains to show that
  $\E{\nu_{\ID}} \leq \eps^2 \cost(L)/2$. For the sake of lightening equations, 
we will denote by $\sum\limits_{\text{bcc. } c}$ the sum over all badly cut 
clients $c$.

  By definition, we have that
  for any solution $S$,
  \begin{align*}
    \cost(S) - \cost_{\calI_{\calD}}(S) &\le
       \sum_{\text{bcc. } c} \dist(c,S)^p - \dist(S,L(c))^p\\
    &\le  \sum_{\text{bcc. } c}
       \big((1+\eps/p)^p\dist(S,L(c))^p \\
    &~~~~~~~~~+ (1+p/\eps)^p\dist(c,L(c))^p - 
    \dist(S,L(c))^p\big)\\
  \end{align*}
  using Lemma~\ref{faketriangleineq} with parameter $x = \eps/p$. 

  To bound $(1+\eps/p)^p$, we use that $\forall x \leq 1/p,~(1+x)^p \leq 
\exp(xp) \leq 1 + px +e (px)^2/2$. Hence, since $\eps \leq 1/4$, 
$(1+\eps/p)^p \leq 1+\eps + e \eps^2 / 2 \leq 1+ 2\eps$. Plugging this into the 
right hand side, we get

\begin{align*}
\cost(S) - \cost_{\calI_{\calD}}(S) 
&\le \sum_{\text{bcc. } c}
       \big((1+2\eps)\dist(S,L(c))^p \\
    &~~~~~~~~~~+ (1+p/\eps)^p \dist(c,L(c))^p - 
    \dist(S,L(c))^p\big)\\
 	& = \sum_{\text{bcc. } c} 2\eps \cdot \dist(S,L(c))^p +
  (1+p/\eps)^p \dist(c,L(c))^p .
\end{align*}
  Subtracting $\sum_{\text{bcc. } c} 2\eps \cdot \dist(S,L(c))^p\leq 
\sum_{c\in C} 2\eps \cdot \dist(S,L(c))^p$ from the right and left side, 
respectively, yields
  \begin{equation*}
    \cost(S) - (1+2\eps) \cost_{\calI_{\calD}}(S)
  \le  \sum_{\text{bcc. } c}(1+p/\eps)^p\dist(c,L(c))^p 
  \end{equation*}

  Similarly, we have that
  \begin{align*}
    \cost_{\calI_{\calD}}(S)- \cost(S) &\le
    \sum_{\text{bcc. } c} \dist(S,L(c))^p - \dist(c,S)^p\\
    &\le \sum_{\text{bcc. } c} \big((1+2\eps) \dist(c,S)^p \\
    &~~~~~~~~~+ (1+p/\eps)^p\dist(c,L(c))^p - \dist(c,S)^p\big)\\
    &\le \sum_{\text{bcc. } c} 2\eps \cdot \dist(c,S)^p +
    (1+p/\eps)^p\dist(c,L(c))^p
  \end{align*}
  and we conclude
  \begin{equation*}
  (1-2\eps)\cost_{\calI_{\calD}}(S)- \cost(S) \le
   \sum_{\text{bcc. } c} (1+p/\eps)^p \dist(c,L(c))^p
  \end{equation*}
   
  Therefore, the expected value of $\nu_{\calI_{\calD}}$ is
  $$
  E[\nu_{\calI_{\calD}}] \le \sum_{\text{client }c} Pr[c \text{ badly cut}] \cdot
  (1+p/\eps)^p \dist(c,L(c))^p.
  $$
  Applying Lemma~\ref{lem:badlycutddim} and using $\badcut = 
\eps^2 (\frac{p}{p+\eps})^p$, we conclude
  $E[\nu_{\calI_{\calD}}] \leq \eps^2\cdot \cost(L)$. The lemma follows
  for a sufficiently small $\eps$. 
\end{proof}

\subsection{Portal Respecting Solution}
In the following, we fix an instance $\I$, a decomposition $\calD$, and a 
solution $L$. By Lemma~\ref{lem:lift}, $\ID$ has small distortion with probability 
at least $1-\eps$ and so we condition on this event from now on. 

We explore the structure that this conditioning can give to solutions.
We will show that in the solution $\opt' = \opt \cup B_{\calD}$ with small cost, 
each client $c$ is cut from its 
serving facility $f$ at a level at most $\log(3L_c/\eps + 4\opt_c)) + \offset$. 
This will allow us to consider \emph{portal-respecting} solution, where every client to
facility path goes in and out parts of the decomposition only at designated portals. Indeed, the
detour incurred by using a portal-respecting path instead of a direct connection depends on the level where 
its extremities are cut, as proven in Lemma~\ref{lem:path-port-resp}. 
Hence, ensuring that this level stays small implies that the detour 
made is small (in our case, $O(\eps(L_c + \opt_c)$).
Such a solution can be computed by a dynamic program that we will present afterwards.

Recall that $L_c$ and $\opt_c$ are the distances from the \emph{original} 
position of $c$ to $L$ and $\opt$, although $c$ 
may have been moved to $L(c)$ and $B_{\calD}$ is the set of badly cut facilities 
of $L$ w.r.t $\calD$.

\begin{lemma}
  \label{lem:detour2}
  Let $\I$ be an instance of Facility Location with a randomized 
decomposition~$\calD$, $\eps < 1/4$ and $L$ be a solution for $\I$, such that $\ID$ has small 
distortion. Let $\opt'=\opt\cup B_\calD$, and for any client $c$ in $\ID$, let 
$\opt'(c)$ be the closest facility to $c$ in $\opt'$. Then $c$ and $\opt'(c)$ 
are cut in $\calD$ at level at most $\log(3L_c/\eps +4\opt_c) + \offset$.
\end{lemma}
\begin{proof}
Let $c$ be a client. To find the level at which $c$ and $\opt'(c)$
  are separated,
  we distinguish between two cases: either $c$ in $\I$ is badly cut w.r.t. 
$\calD$, or not.
    
    \newpage
  If $c$ is badly cut, then it is now located at
  $L(c)$ in the instance $\ID$.
  In that case, either:
 \begin{enumerate}
 \item $L(c)$ is also badly cut, and therefore $L(c) \in 
B_{\calD} \subseteq \opt'$ and so
   $\opt'(c) = L(c)$. Since $c$ and $L(c)$ are collocated, it follows that $c$ and
   $\opt'(c)$ are never cut.
 \item 
 $L(c)$ is not badly cut:
Definition~\ref{def:badlycut} implies that
  $L(c)$ and $\opt(L(c))$ are cut at a level at most $\log(3\opt_{L(c)}) + 
\offset$.
By triangle inequality, $\opt_{L(c)} \leq L_c + 
\opt_c$, and thus 
  $c$ (located at $L(c)$ in $\ID$) and $\opt'(c)$ are also cut at level at most
  $\log \big(3L_c + 3\opt_c\big) + \offset$. 
  \end{enumerate}
 
 We now turn to the case where $c$ is not badly cut. In this case $c$ is not moved to~$L(c)$ and 
the ball $\beta(c, 3L_c/\eps)$ is cut at 
level at most $\log(3L_c/\eps)+\offset$. We make a case distinction according to $\opt_c$ and~$L_c$.
 \begin{enumerate}
  \item If $L_c \leq \eps \opt_c$, then we have the following. If $L(c)$ is 
badly cut, $L(c)$ is open in $\opt'$ and therefore $\opt'_c = L_c$. Moreover, 
since $c$ is not badly cut the ball $\beta(c, L_c)$ is cut at level at most 
$\log(3L_c/\eps) + \offset$. Therefore $c$ and $\opt'(c)$ are cut at level at most 
$\log(3L_c/\eps) + \offset$.

Now consider the case where $L(c)$ is not badly cut. Both $c$ and $\opt'(c)$ 
lie in the ball centered at $L(c)$ and of diameter $2\opt_{L(c)}$: indeed, we use $L_c \leq \eps \opt_c$ to derive
    \begin{align*}
     \dist(c, L(c)) &\leq \eps \dist(c, \opt(c)) \leq \eps \dist(c, \opt(L(c))) 
\\
     &\leq \eps \dist(c, L(c)) + \eps \dist\big(L(c), \opt(L(c))\big),
    \end{align*}
    and therefore $\dist(c, L(c)) \leq \frac{\eps}{1-\eps} \opt_{L(c)} 
\leq 2 \opt_{L(c)}$, since $\eps \leq 1/4$. On the other hand,
    \begin{align*}
     \dist(\opt'(c), L(c)) 
     &\leq \dist(c,\opt'(c)) + \dist(c, L(c)) \\
     &\leq \dist(c,\opt(c)) + \dist(c, L(c)) \\
     &\leq \dist(c, \opt(L(c))) + \dist(c, L(c))\\
     &\leq 2\dist(c, L(c)) + \dist(L(c), \opt(L(c))) \\
     &\leq \left(1+\frac{2\eps}{1-\eps}\right)\opt_{L(c)},
    \end{align*}
    which is smaller than $2\opt_{L(c)}$ for any $\eps \leq 1/4$. Hence we have 
$c, \opt'(c) \in \beta(L(c), 2\opt_{L(c)})$. 

By definition of badly cut, $c$ and $\opt'(c)$ are therefore cut at level at most
$\log(3\opt_{L(c)}) + \offset$. 
Since $\opt_{L(c)} \leq \dist(L(c), \opt(c)) \leq \dist(L(c), c) + 
\dist(c, \opt(c)) \leq (1+\eps)\opt_c$ as $L_c\leq\eps\opt_c$, we have that 
$\log (3\opt_{L(c)}) \leq \log (4\opt_c)$. Hence $c$ and 
$\opt'(c)$ are cut at level at most $\log (4\opt_c) + \offset$.

%
%
%

   \item If $\opt_c \leq L_c / \eps$, then since $c$ is not badly cut the ball $\beta(c, L_c/\eps)$ in which lies $\opt_c$ is cut at level at most $\log (3L_c/\eps) + \offset$.
 \end{enumerate}
 In all cases, $c$ and $L(c)$ are cut at level at most $\log(3L_c/\eps +4\opt_c) + \offset$. This concludes the proof.
\end{proof}

We then aim at proving that there exists a near-optimal
``portal-respecting'' solution, as we define below. A \emph{path} between two nodes $u$ and $v$ is a sequence of nodes $w_1, \ldots, 
w_k$, where $u = w_1$ and $v = w_k$, and its length is $\sum \dist(w_j, w_{j+1})$. A solution can be seen as a set of 
facilities, together with a path for each client that connects it to a facility, 
and the cost of the solution is given by the sum over all path lengths. 
We say that a 
path $w_1, \ldots, w_k$ is \emph{portal-respecting} if for every pair 
$w_j,w_{j+1}$, whenever $w_j$ and $w_{j+1}$ lie in different parts 
$B,B'\in\calB_i$ of the decomposition $\calD$ on some level $i$, then these 
nodes are also portals at this level, i.e., 
$w_j,w_{j+1}\in\calP_B\cup\calP_{B'}$. As explained in Lemma~\ref{lem:path-port-resp}, if two vertices $u$ and $v$ are cut at level $i$, then there exists a portal-respecting path from $u$ to $v$ of length at most $\dist(u,v) + 16\rho 2^i$. 
We define a \emph{portal-respecting} solution to be a solutions such that each 
path from a client to its closest facility in the solution is portal-respecting.

The dynamic program will compute an optimal portal-respecting solution.
Therefore, we need to prove 
that the optimal portal-respecting solution is close to the optimal solution.
We actually show something slightly stronger. Given a solution $S$, we define $b(S) := \sum\limits_{c,i: ~ c \text{ and } S(c) 
\text{ cut at level } i} \eps \growthrate^i$: one can see $b(S)$ as a 
\emph{budget}, given by the fact that vertices are not badly cut. Next we show a structural lemma, that bounds the cost of a structured solution and of its budget.

\begin{lemma}[Structural lemma]
  \label{lem:port-resp}
  Given an instance $\I$, an $\eps$ such that $0 < \eps \leq 1/4$ and a solution $L$, it holds with 
  probability $1 - \eps$  (over~$\calD$) that
  there exists a portal-respecting solution $S$ in $\ID$ such that 
  $\cost_{\ID}(S) + b(S) = (1+O(\eps)) \cost_{\I}(\opt) + O(\eps 
\cost_{\I}(L))$.
\end{lemma}
\begin{proof}
From Lemma~\ref{lem:lift}, with probability $1-\eps$ it holds that the instance 
$\ID$ has small distortion, and $\cost_{\ID}(\opt')\leq 
(1+4\eps)(\cost_\I(\opt)+\eps\cost_\I(L))$.
%
  We now bound the cost of making $\opt'$ portal respecting by applying 
Lemma~\ref{lem:detour2}. Since each client $c$ of $\ID$ is cut from $\opt'(c)$ 
at level at most $\log(3L_c/\eps +4\opt_c) + \offset$, we have that the detour 
for making the assignment of $c$ to $\opt'(c)$ portal-respecting is at most 
$O(\rho \growthrate^{\offset}(L_c/\eps+\opt_{c}))$. Choosing $\rho = \eps^2 
\growthrate^{-\offset}$ ensures that the detour is at most 
$O(\eps(L_c+\opt_{c}))$. 
Summing over all clients $c$ gives a total detour of 
$O(\eps)(\cost_\I(L)+\cost_\I(\opt))$. The resulting portal respecting tour is 
the solution $S$ we are looking for.
The above calculation also bounds $b(S)=b(\opt') \leq 
O(\eps)(\cost_\I(L)+\cost_\I(\opt))$, and so 
$\cost_{\ID}(S) + b(S) = (1+O(\eps)) \cost_{\I}(\opt) + O(\eps \cost_{\I}(L))$.
  %
\end{proof}

\subsection{The Algorithm}\label{sec:alg}
Using Lemma~\ref{lem:asp-rat} and \ref{lem:asp-rat2}, we can assume that the aspect-ratio of 
the instance is $O(n^5/\eps)$. Our algorithm starts by computing a 
constant-factor approximation $L$,
using Meyerson's algorithm~\cite{meyerson2001online}.
It then computes a hierarchical decomposition $\calD$, as explained in the
Section~\ref{sec:decomp}, with parameter~$\rho = \eps \growthrate^{-\offset}$.

Given $L$ and the decomposition $\calD$, our algorithm finds all the badly cut 
clients as follows. For each client~$c$, to determine whether $c$ is badly cut 
or not, the algorithm checks whether the 
decomposition cuts $\beta(c, 3L_c/\eps)$ at a level higher than $\log(3L_c/\eps)+\offset$, making $c$ badly cut. 
This can be done efficiently, since $c$ is in exactly one part at each level, by verifying whether $c$ is at distance smaller 
than $3L_c /\eps$ to such a part of too high level. Thus, the algorithm finds all the badly cut 
clients in near-linear time.

The next step of the algorithm is to compute instance $\ID$ by
moving every badly cut client $c$ to its facility in $L$.
This can be done in linear time.

\paragraph{A first attempt at a dynamic program.}
We now turn to the description of the dynamic program (DP)
for obtaining the best portal-respecting solution of $\ID$. 
This is the standard dynamic program for Facility Location and
we only describe it for the sake of completeness. The reader familiar
with this can therefore skip to the analysis.

%

There is a table entry for each part of the decomposition, and two
vectors of length~$|\calP_B|$, where~$\calP_B$ is the set of portals in 
the part $B$. We call such a triplet a configuration.
Each configuration given by a part $B$ and vectors $\langle 
\ell_1,\ldots,\ell_{|\calP_B|} \rangle$ and $\langle s_1,\ldots,s_{|\calP_B|} 
\rangle$ (called the \emph{portal parameters}), 
encodes a possible interface between part $B$ and a solution for which the 
$i$th portal has distance $\ell_i$ to the closest facility inside of $B$, and 
distance $s_i$ to its closest facility outside of $B$.
The value stored for such a configuration in a table entry is the 
minimal cost for a solution with facilities respecting the constraints induced
by the vectors on the distances between the solution and the portals inside the 
part (as described below).

To fill the table, we use a dynamic program following the lines of
Arora et al.~\cite{AroraRagRao98} or Kolliopoulos
and Rao~\cite{kolliopoulos2007nearly}.
If a part has no descendant (meaning the part contains a single point),
computing the solution given the configuration is straightforward: either a center
is opened on this point or not, and it is easy to check the consistency with the
configuration, where only the distances to portals inside the part need to 
be verified. At a higher level of the decomposition, a solution is simply 
obtained by going over all the sets of parameter values for all the children 
parts. It is immediate to see whether sets of parameter values for the 
children can lead to a consistent solution: 
\begin{itemize}
 \item for each portal $p_1$ of the parent part,
there must be one portal $p_2$ of a child part such that the distance from $p_1$ 
to a center inside the part prescribed by the configuration corresponds to 
$\dist(p_1, p_2)$ plus the distance from $p_2$ to a center inside the child 
part;
 \item for each portal $p_2$ of a child part, there must exist either: 
\begin{itemize}
 \item a portal $p_1$ of the parent part such that the distance from
$p_2$ to a center outside its part prescribed by the configuration is $\dist(p_1, p_2)$
plus the distance from $p_1$ to a center outside of the part,
\item or a portal $p_1$ of another child part
such that this distance is $\dist(p_1, p_2)$
plus the distance from $p_1$ to a center inside the child part.
\end{itemize}
\end{itemize}

The runtime of this algorithm depends on the number of possible distances 
determining the number of possible portal parameters. Even if the aspect ratio 
is polynomial, there can be a large number of possible distances, so that the 
number of configurations might be exponential. Using the budget given by 
Lemma~\ref{lem:port-resp}, one can approximate the distances and obtain an efficient 
algorithm, as we show next.

\paragraph{A faster dynamic program.}
We now describe a faster dynamic program.
Consider a level where the diameter of the parts is say $D$. 
Each configuration is again given by a part $B$ and portal parameters $\langle 
\ell_1,\ldots,\ell_{|\calP_B|} \rangle$ and $\langle s_1,\ldots,s_{|\calP_B|} 
\rangle$, 
but with the restriction that $\ell_i$ and $s_i$ are multiples of $\eps D$ 
in the range $[0, D/\eps + D]$. A boolean flag is additionally attached to the 
configuration (whose meaning will be explained shortly). 

We sketch here the intuition behind this restriction. Since the diameter of the 
part is $D$ we can afford a detour of $\eps D$, that can be charged to the 
budget $b(S)$. Hence, distances can be rounded to the closest multiple of $\eps 
D$. 

Now, suppose that the closest facility outside the part is at distance greater 
than $D/\eps$, and that there is no facility inside the part. Then, since the 
diameter is $D$, up to losing an additive $D\leq \eps\opt$ in the cost of 
the solution computed, we may assume that all the points of the part are 
assigned to the same facility. So the algorithm is not required to have the 
precise distance to the closest center outside the part, and it uses the flag to 
reflect that it is in this regime. We can then treat this whole part as a single 
client (weighted by the number of clients inside the part) to be considered at 
higher levels. Assuming that the closest facility is at distance less than $D / 
\eps$, we have that for any portal of the part the closest facility is at 
distance at most $D / \eps + D$ (since $D$ is the diameter of the part).

On the other hand, if there is some facility inside the part
and the closest facility outside the part is at distance
at least $D / \eps$, then each client of the part should be
served by a facility inside the part in any optimal assignment.
Thus it is not necessary that the algorithm iterates over configurations 
where the distances outside the part are more than~$D / \eps $: it is enough 
to do it once and use the answer for all other queries.

\paragraph{Analysis -- Proof of Theorem~\ref{thm:fl}.} 

The following lemmas show that the solution computed by this algorithm is
a near-optimal one, and that the complexity is near-linear: this proves 
Theorem~\ref{thm:fl}. We first bound the connection cost in $\ID$.

\begin{lemma}\label{lem:alggood}
 Let $S$ be as in Lemma~\ref{lem:port-resp}. The algorithm computes a solution $S^*$ with
 cost at most $\cost_{\ID}(S^*) \leq (1+O(\eps))\cost_{\ID}(S) + b(S)$.
\end{lemma}
\begin{proof}
We show that the solution $S$ can be adapted to a configuration of the DP with 
extra cost $b(S)$. For this, let $c$ be a client served by a facility 
$S(c)$, and let $w_1,\ldots,w_k$ be the portal-respecting path from $c$ to 
$S(c)$ with $w_1 = c$ and $w_k = S(c)$. The cost contribution of $c$ to $S$ is 
therefore $\sum_{i=1}^{k-1} \dist(w_i, w_{i+1})$. For each $w_i$, let also $l_i$ 
be the level at which $w_i$ is a portal.

The distance between $c$ and $S(c)$ is approximated at several places of the 
DP. Consider any node $w_i$ on the path from $c$ to $S(c)$:
\begin{itemize}
 \item When $\dist(w_i, S(c)) \leq 2^{l_i}/\eps + 2^{l_i}$, the distance 
between $w_i$ and $S(c)$ is rounded to the closest multiple of $\eps 2^{l_i}$, 
incurring a cost increase of $\eps 2^{l_i}$.
 \item When $\dist(w_i, S(c)) \geq 2^{l_i}/\eps + 2^{l_i}$, the whole part is contracted
 and served by a single facility at distance at least $2^{l_i}/\eps$. The cost 
 difference for client $c$ is therefore~$2^{l_i}$.
 Since the diameters of the parts are geometrically increasing, the total cost
 difference for all contractions of regions containing $c$ is bounded by $2^{l_j + 1}$, where $l_j$ is
 the highest level where $d(w_j, S(c)) \geq 2^{j_i}/\eps + 2^{l_j}$. This 
inequality implies that $2^{l_j + 1} \leq 2\eps d(w_j, S(c))$, which is smaller 
than  $2 \eps \sum d(w_i, w_{i+1})$, the cost of $c$ in the portal-respecting 
solution~$S$.
\end{itemize}

Hence, summing over all clients, the additional cost incurred by the DP is 
at most $b(S) + 4 \eps \cost_{\ID}(S)$.
Since it computes a solution with minimal cost, it holds that 
$\cost_{\ID}(S^*) \leq (1+4\eps)\cost_{\ID}(S) + b(S)$.
\end{proof}

Next we bound the connection cost in $\I$. If $\ID$ has small distortion, the 
facility cost increase due to badly cut clients is bounded by $\sum_{f\in 
B_\calD} w_f\leq \eps \sum_{f\in L} w_f$, since we have $\opt'=\opt\cup 
B_{\calD}$. Thus due to the following corollary the total solution cost of $S^*$ 
is bounded as required.

\begin{corollary}
 Let $S^*$ be the solution computed by the algorithm. With probability $1-\eps$,
 it holds that $\cost_{\I}(S^*) = (1+O(\eps)) \cost_{\I}(\opt)$.
\end{corollary}
\begin{proof}
 Lemma~\ref{lem:port-resp} ensures that, with probability $1-\eps$, the cost of $S$ 
in $\ID$ and $b(S)$ is at most
$(1+O(\eps)) \cost_{\I}(\opt) + O(\eps \cost_{\I}(L))$. Since $L$ is a 
constant-factor approximation of $\opt$ in $\I$, this cost turns out to be 
$(1+O(\eps))\cost_{\I}(\opt)$. Using that $\ID$ has small distortion, and 
combining this with Lemma~\ref{lem:alggood} concludes the proof:
  \begin{align*}
    \cost_{\I}(S^*) &\leq (1+2\eps)\cost_{\ID}(S^*)+\eps\cost_\I(L)\\
    &\leq (1+O(\eps))(\cost_{\ID}(S) + b(S))+\eps\cost_\I(L)\\
    &\leq (1+O(\eps)) \cost_{\I}(\opt) \qedhere
  \end{align*}
\end{proof}

\begin{lemma}
 This algorithm runs in $2^{(1/\eps)^{O(d^2)}}n+2^{O(d)}n\log n$ time.
\end{lemma}
\begin{proof}
The preprocessing step (computing $L$, the hierarchical decomposition $\calD$, 
and the
instance~$\ID$) has a running time $O(n\log n)$, as all the steps can be done 
with this complexity: a fast implementation of Meyerson's algorithm~\cite{meyerson2001online}
tailored to graphs 
can compute $L$ in time $O(m \log n)$. Using it on the spanner with constant 
stretch computed with
\cite{har2006} gives a $O(1)$-approximation in time $O(n\log n)$. As explained
earlier, the hierarchical decomposition $\calD$ and the instance $\ID$ can also 
be computed with this complexity. The decomposition can moreover be transformed in order to remove part that do not contain any point, as well as degree $2$ nodes. This ensures to have $O(n)$ part in total, since there are $n$ leaves and a degree at least $3$.

The DP has a linear time complexity: in a part of diameter $D$, the 
portal set is a $(\eps \growthrate^{-\offset} D)$-net, and hence has size 
$2^{d \lceil\log (\growthrate^{\offset}/\eps )\rceil}$ by Lemma~\ref{prop:doub:net}. 
Since $\offset = 2d+2 +\log \frac{(p+\eps)^p}{\eps^2p^p}$, 
this number can be simplified to $2^{O(d^2+d\log(1/\eps))}$. Since each portal stores a distance that can take only $1/\eps^2$ 
values, there are at most $T=(1/\eps^2)^{2^{O(d^2+d\log(1/\eps))}} = 
2^{2^{O(d^2+d\log(1/\eps))}}$ possible table entries for a given part. 

To fill the table, notice that a part has at most $2^{O(d)}$ children, due to 
the properties of the hierarchical decomposition. For any given part, going over all the sets of parameter values for 
all the children parts therefore takes time $T^{2^{O(d)}} = 
2^{2^{O(d^2+d\log(1/\eps))}}$. This dominates the complexity of computing all table entry for one part of the decomposition.

Since the hierarchical decomposition is a tree with $n$ leaves (one per vertex) and without degree-one internal nodes (those can be compressed), there are at most $n$ parts in the decomposition: the complexity of the dynamic program is therefore $n\cdot 2^{2^{O(d^2+d\log(1/\eps))}}$, 

The total complexity of the algorithm is thus
\[n\cdot 2^{2^{O(d^2+d\log(1/\eps))}} + 2^{O(d)}n\log n \qedhere\]
\end{proof}

\section{The $k$-Median and $k$-Means Problems}
\label{sec:structkmedian}

We aim at using the same approach as for Facility Location. 
We focus the presentation on $k$-Median, and only later show how to adapt
the proof for $k$-Means.

We will work with the more general version of $k$-Median as defined in 
Section~\ref{sec:preprocess}, where the instance consists of a set of clients $C$,
a set of candidate centers $F$, an integer $k$, and a function $\chi : C \mapsto 
\{1,\ldots,n\}$ and the goal is to minimize  $\sum_{c \in C} \chi(c) \cdot 
\min_{f \in S} \dist(c,f)$. We will consider in the proof that $C$ is actually a multiset, instead of carrying along the multiplicity $\chi$.

The road-map is as for Facility Location: we show in Lemma~\ref{prop:structkmedian} 
that an instance $\ID$ has a \emph{small distortion} with good probability, and 
then in Lemma~\ref{lem:detourkmedians} that if an instance has small distortion then 
there exists a near-optimal portal-respecting solution. We finally present a 
dynamic program that computes such a solution.

A key ingredient of the proof for Facility Location was our ability to add all 
badly-cut facilities to the solution $\opt'$. This is not directly possible in 
the case of $k$-Median and $k$-Means, as the number of facilities is fixed. 
Hence, the first step of our proof is to show that one can make some room in 
$\opt$, by removing a few centers without increasing the cost by too much.

\subsection{Towards a Structured Near-Optimal Solution}
Let $\globalS$ be an optimal solution to $\I$ and $L$ and approximate solution. 
We consider the mapping of the 
facilities of $\globalS$ to $L$ defined as follows: for 
any $f \in \globalS$, let $L(f)$  denote the facility of $L$ that is the closest 
to $f$. Recall that for a client $c$, $L(c)$ is the facility serving $c$ in $L$.

For any facility $\ell$ of~$L$, define $\psi(\ell)$ to be the set of facilities 
of $\globalS$ that are mapped to $\ell$, namely, $\psi(\ell) = \{f \in \opt \mid 
L(f) = \ell \}$. Define $L^1$ to be the set of facilities $\ell$ of $L$ for 
which there exists a unique $f \in \globalS$ such that $L(f) = \ell$, namely 
$L^1 = \{\ell\in L \mid |\psi(\ell)| = 1\}$. Let $L^0 = \{ \ell\in L \mid 
|\psi(\ell)| = 
0\}$, and $L^{\geq 2} = L \setminus(L^1 \cup L^0)$. Similarly, define $\globalS^1 = \{ 
f \in \globalS \mid L(f) \in L^1\}$ and $\globalS^{\geq 2} = \{ f \in \globalS \mid 
L(f) \in L^{\geq 2}\}$. Note that $|\globalS^{\geq 2}| = |L^0|+|L^{\geq 2}|$, since 
$|\globalS^1|=|L^1|$ and, w.l.o.g., $|\globalS|=|L|=k$.

\medskip
The construction of a structured near-optimal solution is made in 3 steps. The 
first one defines a solution $\opt'$ as follows.
Start with $\opt' = \globalS$.
\begin{itemize}
\item \textbf{Step 1.}
For each facility $\ell\in L^{\geq 2}$, fix one in $\globalS^{\geq 2}$ that is 
closest to~$\ell$, breaking ties arbitrarily, and call it $f_\ell$. Let 
$\calH\subseteq\globalS^{\geq 2}$ be the set of facilities of $\globalS^{\geq 
2}$ that are not the closest to their corresponding facility in $L^{\geq 2}$, 
i.e., $f\in\calH$ if and only if $f\in\psi(\ell)$ and $f\neq f_\ell$ for some 
$\ell\in L^{\geq 2}$. Among the facilities of $\calH$, remove from $\opt'$ the 
subset of size $\lfloor \eps\cdot  |\globalS^{\geq 2}| / 2\rfloor$ that yields 
the smallest cost increase. Note that this subset is well-defined if $\eps\leq 
1$.
\end{itemize}
This step makes room to add the badly cut facilities without violating the 
constraint on the maximum number of centers, while at the same time ensures 
near-optimal cost, as the following lemma shows.

\begin{lemma}\label{lem:coststep1}
 After Step 1, $\opt'$ has cost $(1+O(\eps))\cost(\globalS) + O(\eps)\cost(L)$
\end{lemma}
\begin{proof}
We claim that for a client $c$ served by $f\in\calH$ in the optimum solution 
$\globalS$, i.e., $f=\opt(c)$, the detour entailed by the deletion of $f$ is 
$O(\globalS_c + L_c )$. Indeed, let $f'$ be the facility of $\globalS$ that is 
closest to $L(f)$, and recall that $L(c)$ is the facility that serves $c$ in the 
solution $L$. Since $f'\notin\calH$, the cost to serve $c$ after the removal of 
$f$ is at most $\dist(c, f')$, which can be bounded by $\dist(c, f') \leq 
\dist(c, f) + \dist(f, L(f)) + \dist(L(f), f')$. But by definition of~$f'$,  
$\dist(f', L(f))\leq\dist(L(f), f)$, and by definition of the function $L$ we 
have $\dist(L(f), f) \leq \dist(L(c), f)$, so that $\dist(c, f') \leq \dist(c, 
f) + 2\dist(f, L(c))$. Using the triangle inequality finally gives $\dist(c, f') 
\leq 3\dist(c, f) + 2\dist(c, L(c))$ which is $O(\globalS_c + L_c)$. For a 
facility $f$ of $\globalS$, we denote by $C(f)$ the set of clients served 
by~$f$, i.e. $C(f) = \{ c \in C \mid \opt(c)~=~f\}$. The total cost incurred by 
the removal of $f$ is then $\sum_{c\in C(f)} O(\opt_c+L_c)$, and the cost of 
removing all of $\calH$ is $O(\cost(\globalS)+\cost(L))$.

Recall that in Step 1 we remove the set $\widehat{\calH}$ of size 
$\lfloor\eps|\globalS^{\geq 2}|/2\rfloor$ from $\calH$, such that 
$\widehat{\calH}$ minimizes the cost increase. We use an averaging argument to 
bound the cost increase: the sum among all facilities $f\in\calH$ of the cost of 
removing the facility $f$ is less than $O(\cost(\globalS)+\cost(L))$, and 
$|\calH| = O(1/\eps) \cdot \lfloor\eps|\globalS^{\geq 2}|\rfloor$. Therefore 
removing~$\widehat{\calH}$ increases the cost by at most 
$O(\eps)(\cost(\globalS)+\cost(L))$, so that Step 1 is not too expensive.
\end{proof}

We can therefore use this solution $\opt'$ as a proxy for the optimal solution,
and henceforth we will denote this solution by $\opt$. In particular, the badly 
cut facilities are defined for this solution and not the original $\opt$.

\subsection{An instance with small distortion}
As in Section~\ref{sec:struct}, the algorithm computes 
a randomized hierarchical decomposition $\calD$, 
and transforms the instance of the problem:
every badly cut client $c$ 
is moved to $L(c)$, namely, there is no more client at $c$ and we
add an extra client at $L(c)$. Again,
we let $\ID$ denote the resulting instance and note that $\ID$
is a random variable that depends on the randomness of~$\calD$.

Moreover, similar as for Facility Location, we let $B_{\calD}$ be the set of 
centers of $L$ that are badly cut from $\opt$, i.e., $f\in B_{\calD}$ if
the ball $\beta(f,3\opt_f)$ is cut at some level greater than $\log(3\opt_f)+\tau(\eps,d)$. 
We call $\cost_{\I}(S)$ the cost of a solution $S$ in the original instance 
$\I$, and $\cost_{\ID}(S)$ its cost in $\ID$. We let 
\[\nu_{\ID} = \max_{\text{solution }S} \big\{\cost_\I(S) -  (1+2\eps) 
\cost_{\ID} (S),  (1-2\eps) \cost_{\ID} (S) - \cost_\I(S)\big\}.\] 
We say that an instance $\ID$ has \emph{small distortion} if $\nu_{\ID} \leq 
\eps \cost_\I(L)$, and there exists a solution $S$ that contains $B_{\calD}$ 
with $\cost_{\ID}(S) \le (1+O(\eps))\cost_{\I}(\opt) + O(\eps) \cost_{\I}(L)$. 
That is, the condition is the same as for Facility Location, except that we do 
not need a bound on the opening costs. In contrast to the Facility Location 
problem, here we need to be more careful when identifying the solution 
fulfilling the latter inequality.

For this, we go on with the next two steps of our construction, defining a 
solution $S^*$. 
Recall that we defined $f_\ell\in\globalS^{\geq 2}$ to be the closest facility 
to $\ell\in L^{\geq 2}$, breaking ties arbitrarily. For any $\ell\in L^1$, we 
also denote by $f_\ell\in\globalS^1$ the unique facility closest to~$\ell$. We 
start with $S^* = \opt$ obtained from Step 1. Note that for every $\ell\in 
L^1\cup L^{\geq 2}$ the closest facility $f_\ell\in \globalS$ is still present after Step 1, 
since only some of the other facilities in $\calH$ were removed.

  \begin{itemize}
  \item \textbf{Step 2.} 
    For each badly-cut facility $\ell \in B_\calD\setminus L^0$ (i.e., 
$\psi(\ell)\neq\emptyset$), replace $f_\ell$ by $\ell$ in~$S^*$.
  \item \textbf{Step 3.}
    Add all badly cut facilities of $L^0$ to $S^*$.
  \end{itemize}
  
  \newpage
We show next that $S^*$ satisfies the conditions for $\ID$ to
have small distortion with good probability.

\begin{lemma}
  \label{prop:structkmedian}
  The probability that $\ID$ has small distortion is
  at least $1-\eps$, if $\eps\leq 1/4$.
\end{lemma}
\begin{proof}
  The proof that $\nu_{\ID} \leq \eps \cost_\I(L)$ with probability
  at least $1-\eps/2$ is identical
  to the one in Lemma~\ref{lem:lift}. We thus turn to bound the probability
  that solution $S^*$ satisfies the cardinality and cost requirements. Our goal 
is to show that
  this happens with probability at least $1-\eps/2$.
  Then, taking a union
  bound over the probabilities of failure yields the proposition.

  By Steps 2 and 3, we have that $S^*$ contains $B_{\calD}$.
  We split the proof of the remaining properties into the following claims.
  
\begin{claim}\label{lem:median:admiss}
With probability at least $1-\eps/4$, the set $S^*$ is an admissible solution, 
i.e., $|S^*| \le k$.
\end{claim}
\begin{proof}
  We let $b$ be the number of facilities of $L^0$ that
  are badly cut. By Lemma~\ref{lem:badlycutddim}, we have that
  $\E{b} \le \eps^2 |L^0|/4$ as $p=1$. By Markov's inequality, the
  probability that $b >\lfloor\eps |L^0|/2\rfloor$ is at most~$\eps/2$.
  Now, condition on the event that $b \le\lfloor\eps |L^0|/2\rfloor$.
  Since $|L^0|+|L^{\geq 2}| = |\globalS^2|$, we have
  that $b \le \lfloor\eps|\globalS^{\geq 2}|/2\rfloor$. Moreover, the three 
steps converting $\opt$ into $S^*$ ensure that  $|S^*| \le k + b -  \lfloor\eps 
|\globalS^{\geq 2}|/2 \rfloor$, as Step 1 removes $\lfloor\eps |\globalS^{\geq 
2}|/2 \rfloor$ facilities, while Step 2 only swaps facilities so their number 
does not change, and Step 3 adds $b$ facilities. Combining the two inequalities 
gives $|S^*| \le k$.
 \end{proof}

\begin{claim}\label{lem:median:cost}
  If $\eps < 1/4$, then with probability at least $1-\eps/4$, $\cost_{\ID}(S^*) \le (1+O(\eps)) 
\cost_\I(\globalS) + O(\eps\cdot\cost_\I(L))$
\end{claim}
\begin{proof}
We showed in Lemma~\ref{lem:coststep1} that the cost increase in $\I$ due to Step~1 
is at most $O(\eps)(\cost_\I(\globalS)+\cost_\I(L))$. We will prove below that 
this implies that also Step~2 leads to a cost increase of 
$O(\eps)(\cost_\I(\globalS)+\cost_\I(L))$ in $\I$ with good probability. Step~3 
can only decrease the cost. Hence we have $\cost_{\I}(S^*) \le (1+O(\eps)) 
\cost_\I(\globalS) + O(\eps\cdot\cost_\I(L))$. To bound the cost of $S^*$ in 
$\ID$, we use that $(1-2\eps)\cost_{\ID}(S^*)-\cost_\I(S^*)\leq 
\nu_{\ID}\leq \eps\cost_\I(L)$. The cost incurred by connecting clients to 
facilities in $S^*$ is 
\begin{align*}
\cost_{\ID}(S^*)&\leq 
(1+\frac{2\eps}{1-2\eps})(\cost_\I(S^*)+\eps\cdot\cost_\I(L))
&\text{(as $\nu_{\ID} \leq \eps \cost_\I(L)$)}
\\
&\leq (1+4\eps)(\cost_\I(S^*)+\eps\cdot\cost_\I(L))
&\text{(as $\eps\leq 1/4$)}
\\
&\leq (1+O(\eps))\cost_\I(\opt)+O(\eps\cdot\cost_\I(L))
&\text{(by above bound on $\cost_\I(S^*)$)}.
\end{align*}

To bound the cost increase of Step~2, we first show that starting with $\opt$ 
and replacing every $f_\ell\in \opt$ by $\ell\in L^1\cup L^{\geq 2}$ results in 
a solution $S'$ of cost $O(\cost_\I(\globalS) + \cost_\I(L))$. For that, let $c$ 
be a client that in $\opt$ is served by a facility $f_\ell$ that is closest to 
some $\ell\in L^1\cup L^{\geq 2}$. Recall that every facility of $\globalS$ that 
is closest to some facility of $L^1\cup L^{\geq 2}$ is in $\opt$, as only some 
of those from $\calH$ are removed in Step 1. Hence if $c$ is served by some 
$\ell'\in L^1\cup L^{\geq 2}$ in the solution~$L$, then this facility $\ell'$ is 
in $S'$ since it will replace the closest facility $f_{\ell'}$. Thus the cost of 
serving $c$ in $S'$ is $\dist(c, \ell')=\dist(c, L)$. On the other hand, if $c$ 
is served by a facility $\ell_0$ of $L^0$ in $L$, then it is possible to serve 
it by the facility $\ell$ that replaces~$f_\ell$. The serving cost then is 
$\dist(c, \ell) \leq \dist(c, f_\ell) + \dist(f_\ell, \ell) \leq \dist(c, 
f_\ell) + \dist(f_\ell, \ell_0)$, using that $f_\ell$ is the closest facility to 
$\ell$ in the last inequality. Using again the triangle inequality, this cost is 
at most $2\dist(c, f_\ell) + \dist(c, \ell_0)$. Moreover, any client served by 
a facility of $\calH$ in $\opt$, i.e., which is not the closest to a facility 
of $L$, can in $S'$ be served by the same facility as in $\opt$, with cost 
$\dist(c, \opt)$. Hence the cost of the obtained solution is at most 
$2\cost_\I(\opt) + \cost_\I(L)\leq O(\cost_\I(\globalS)+\cost_\I(L))$ by 
Lemma~\ref{lem:coststep1} and assuming, say, $\eps\leq 1$.

The probability of replacing $f_\ell$ by $\ell\in L^1 \cup L^{\geq 2}$ in Step~2 
is the probability that $\ell$ is badly cut. This is $\badcut$ by 
Lemma~\ref{lem:badlycutddim} (note that this probability is the same whether 
$B_\calD$ is defined for $\opt$ or $\opt$). Finally, with linearity of 
expectation, the expected cost to add the badly cut facilities $\ell\in L^1\cup 
L^{\geq 2}$ instead of their closest facility $f_\ell$ of $\opt$ in Step~2 is 
$O(\badcut (\cost_\I(\globalS) + \cost_\I(L)))$. Markov's inequality thus 
implies that the cost of this step is at most 
$O(\eps)(\cost_\I(\globalS)+\cost_\I(L))$ with probability 
$1-\frac{O(\badcut)}{\eps} \geq 1 - \eps/4$, since $ \badcut \leq \eps^2/2$ in 
the case of $k$-Median where $p=1$, and $\eps < 1/4$.
\end{proof}

%
%

Lemma~\ref{prop:structkmedian} follows from taking a union bound over the 
probabilities of failure of Claim~\ref{lem:median:admiss} and~\ref{lem:median:cost}.
\end{proof}

\subsection{Portal respecting solution}

We have to prove the same structural lemma as for Facility Location, to say that 
there exists a portal-respecting solution in $\ID$ with cost close to 
$\cost(S^*)$ where $S^*$ is the solution obtained from the three steps above.
Recall that for any solution $S$ and client $c$, $S_c$ is the distances from the 
\emph{original} position of $c$ to $S$ in $\I$, but $c$ may have been moved to 
$L(c)$ in~$\ID$.
Recall also that $\opt$ is defined after removing some centers in Step 1.

\begin{lemma}
  \label{lem:detourkmedians}
Let $\I$ be an instance of $k$-Median with a randomized decomposition $\calD$, 
and $L$ be a solution for $\I$, such that $\ID$ has small distortion. Let 
$S^*$ be the solution obtained by applying Steps~1,~2 and~3. Then, for any client $c$ in 
$\ID$,  $c$ and $S^*(c)$ are cut at level at most $\log(4\opt_c + 3L_c/\eps) + \offset$ in $\calD$,
whenever $\eps\leq 1/5$, where $S^*(c)$ is the closest facility to $c$ in $S^*$.
\end{lemma}
\begin{proof}
The proof of this lemma is very similar to the one of Lemma~\ref{lem:detour2}. 
However, since some facilities of $\opt$ were removed in Step 2 to obtain $S^*$, 
we need to adapt the proof carefully. In particular, we will use the following 
inequality. Let $c$ be a client. If $\opt(c)$ was moved in Step~2, it was 
replaced by facility $\ell$ for which $\opt(c)=f_\ell$ and $\dist(\opt(c), \ell) 
\leq \dist(\opt(c), L(c))$, since  $f_\ell$ is the closest facility to 
$\ell$. Using the triangle inequality we obtain $\dist(\opt(c), L(c))\leq 
\dist(c, \opt(c)) + \dist(c, L(c))$. On the other hand, as $\ell\in S^*$ we get 
$\dist(c,S^*(c))\leq \dist(c,\ell)\leq \dist(c, \opt(c))+\dist(\opt(c), \ell)$, 
again applying the triangle inequality. Putting these inequalities together we 
obtain
\begin{equation}\label{prop-star}
\dist(c,S^*(c))\leq 2\dist(c, \opt(c)) + \dist(c, L(c)).
\end{equation}
Furthermore, if $\opt(c)$ is not moved in Step~2 we have $\opt(c)\in S^*$, and 
so 
Inequality~\eqref{prop-star} holds trivially as $\dist(c,S^*(c))\leq \dist(c, 
\opt(c))$.

To find the level at which $c$ and $S^*(c)$ are cut, we distinguish 
between two cases: either $c$ in $\I$ is badly cut w.r.t.~$\calD$, or 
not.
If $c$ is badly cut, then it is now located at $L(c)$ in the instance $\ID$.
In that case, either:

\begin{enumerate}
 \item $L(c)$ is also badly cut, i.e., $L(c) \in 
B_\calD\subseteq S^*$, and so $S^*(c) =L(c)$. It follows that $c$ and 
$S^*(c)$ are collocated, thus they are never cut.
   
 \item $L(c)$ is not badly cut. 
 Then, since $c$ is now located at $L(c)$, $\dist(c, S^*(c)) = \dist(L(c), S^*(L(c))$. $\opt(L(c))$ is 
not necessarily in $S^*$: in that case, it was replaced by a facility $f$ that is closer to $\opt(L(c))$ than $L(c)$, and so $d(L(c), f) \leq 2d(L(c), \opt(L(c))$. Hence, either if $\opt(L(c))$ is in $S^*$ or not, it holds that $d(L(c), S^*) \leq 2 \opt_{L(c)}$.
 
 Since $L(c)$ is not badly cut, the ball $\beta(L(c), 3\opt_{L(c)})$
 is cut at level at most $\log (3\opt_{L(c)}) + \offset$.
  By triangle inequality, $\opt_{L(c)} = \dist(L(c), \opt(L(c))) \leq L_c + 
\opt_c$, and thus 
  $c$ and $S^*(c)$ are also separated at level at most
  $\log \big(3L_c + 3\opt_c\big) + \offset$.
  \end{enumerate}
 
 \medskip
 In the other case where $c$ is not badly cut,  
 the ball $\beta(c, 3L_c/\eps)$ is 
cut at level at most $\log(3L_c/\eps)+\offset$. We make a case distinction according to $L_c$ and $\opt_c$.
 \begin{enumerate}
  \item
If $L_c \leq \eps \opt_c$, then we have the following. 
    If $L(c)$ is badly cut, $L(c)\in B_\calD\subseteq S^*$ and therefore 
$S^*_c \leq L_c$. Moreover, since $c$ is not badly cut the ball $\beta(c, L_c)$ 
is cut at level at most $\log(3L_c/\eps) + \offset$. Therefore $c$ and $S^*(c)$ are 
cut at a level below $\log(4\opt_c + 3L_c/\eps) + \offset$.
    
    In the case where $L(c)$ is not badly cut, both $c$ and $S^*(c)$ lie in the 
ball centered at $L(c)$ and of diameter $3\opt_{L(c)}$, which can be seen as 
follows. We use $L_c \leq \eps \opt_c$ to derive
    \begin{align*}
     \dist(c, L(c)) &\leq \eps \dist(c, \opt(c)) \leq \eps \dist(c, \opt(L(c))) 
\\
     &\leq \eps \dist(c, L(c)) + \eps \dist\big(L(c), \opt(L(c))\big)
    \end{align*}
    And therefore, since $\eps \leq 1/4$, $\dist(c, L(c)) \leq 
\frac{\eps}{1-\eps} \opt_{L(c)} \leq \opt_{L(c)}/3$. 

Using these inequalities 
we also 
get
    \begin{align*}
     \dist(S^*(c), L(c)) &\leq \dist(S^*(c), c) + \dist(c, L(c)) \\
      &\leq 2\dist(c, \opt(c)) + 2\dist(c, L(c)) 
      &\text{ (using Inequality~\eqref{prop-star})}\\
      &\leq 2\dist(c, \opt(L(c))) + 2\dist(c, L(c))\\
      &\leq 4\dist(c, L(c)) + 2\dist(L(c), \opt(L(c))) \\
      & \leq \left(2+\frac{4\eps}{1-\eps}\right)\opt_{L(c)},
    \end{align*}
    which is smaller than $3\opt_{L(c)}$ for any $\eps \leq 1/5$. Hence we have 
$c, S^*(c) \in \beta(L(c), 3\opt_{L(c)})$. Since $L(c)$ is not badly cut, $c$ and $S^*(c)$ are cut at level at most $\log(3\opt_{L(c)})+\offset$. Since $\dist(L(c), \opt(L(c))) \leq 
\dist(L(c), \opt(c)) \leq \dist(L(c), c) + \dist(c, \opt(c)) \leq 
(1+\eps)\opt_c$, we 
have that $\log(3\opt_{L(c)})+\offset \leq \log (4\opt_{L(c)}) + \offset$.

  \item If $\opt_c \leq L_c / \eps$, then $2\opt_c + L_c \leq 3L_c / \eps$ and since $c$ is not badly
    cut, the ball $\beta(c, 2\opt_c + L_c)$ is cut at level at most 
    $\log (3L_c / \eps) + \offset$. 
    Moreover, $S^*(c)$ lies in this ball by Inequality~\eqref{prop-star}.
  
%
 \end{enumerate}
In all cases, $c$ and $S^*(c)$ are separated at level at most $\log(3L_c/\eps + 4\opt(c)) + \offset$, which concludes the lemma.
\end{proof}

Equipped with these two lemmas, we can prove the following lemma, which 
concludes the section.
Note again, that the bounds are for $\opt$ defined after removing some centers 
in Step 1.

\begin{lemma}
  \label{lem:port-resp-medians}
  Condition on $\ID$ having small distortion. There exists a portal-respecting 
solution $S$ in $\ID$ such that $\cost_{\ID}(S) + b(S) \le (1+O(\eps)) 
\cost_{\I}(\opt) + O(\eps \cost_{\I}(L))$.
\end{lemma}
\begin{proof}
The proof follows exactly the one of Lemma~\ref{lem:port-resp}, making $S^*$ 
portal-respecting, and using Lemma~\ref{prop:structkmedian} and 
Lemma~\ref{lem:detourkmedians} to prove that the resulting solution $S$ in $\ID$ has
$\cost_{\ID}(S) + b(S) \le (1+O(\eps)) \cost_{\I}(\opt) + O(\eps 
\cost_{\I}(L))$.
\end{proof}

\paragraph{Extension to $k$-Means}
The adaptation to $k$-Means -- and more generally $k$-Clustering -- can be essentially captured by the following 
inequality: $(x+y)^p \leq 2^p(x^p + y^p)$. Indeed, taking the example of 
Claim~\ref{lem:median:cost}, the detour $\dist(c, f') \leq 3\dist(c, f) + 2\dist(c, 
l)$ gives a cost $\dist(c, f')^p = O(\dist(c, f)^p + \dist(c, l)^p)$. This follows through 
all the other lemmas, and therefore the above lemmas also hold for $k$-Means 
with larger constants.
\newpage
\subsection{The Algorithm}


The algorithm follows the lines of the one for Facility Location, in Section~\ref{sec:alg}. It first computes a constant-factor approximation $L$, then the hierarchical decomposition $\calD$ (with parameter $\rho = \eps \growthrate^{-\offset}$) and constructs instance $\ID$. A dynamic program is then used to solve efficiently the problem, providing a solution $S$
of cost at most $(1+\eps)\cost_{\I}(\opt)$ -- conditioned on the
event that the instance $\ID$ has
small distortion.

\paragraph{Dynamic programming.}
The algorithm proceeds bottom up along the levels of
the decomposition.
We give an overview of the dynamic program which is a slightly
refined version of the one presented for Facility Location
in Section~\ref{sec:alg}. We make use of two additional ideas.

To avoid the dependency on $k$ we proceed as follows.
In the standard approach, a cell of the dynamic program is defined by
a part of the decomposition $\calD$,
the portal parameters (as defined in Section~\ref{sec:alg}),
and a value
$k_0 \in [k]$. The value of an entry in the table
is then the cost of the best solution that uses $k_0$ centers,
given the portal parameters.

For our dynamic program for the $k$-Median and $k$-Means problems,
we define a cell of the dynamic program by a part~$B$, the portal
parameters $\langle 
\ell_1,\ldots,\ell_{n_p} \rangle$ and $\langle s_1,\ldots,s_{n_p} \rangle$ and a value $c_0$ in $[\cost(L)/n; (1+\eps)\cost(L)]$. The entry
of the cell is equal to the minimum number $k_0$ of centers that
need to be placed in part $B$ in order to achieve cost at most $c_0$,
given the portal parameters.
Moreover, we only consider values for $c_0$ that are
powers of $(1+\eps/\log n)$. The output of the algorithm is the minimum value 
$c_0$ such that the root cell has value at most $k$ (i.e., the minimum value 
such that at most $k$ centers are needed to achieve it).


The DP table can be computed the following way. For the parts that
have no descendant, namely the base cases,
computing the best clustering given
a set of parameters can be done easily: there is at most one client in
the part,
and verifying that the parameter values for the centers inside the part
are consistent can be done easily.
At a higher level of the decomposition, a solution is obtained
by going over all the sets of parameter values for all the children
parts. It is immediate to see whether sets of
parameter values for the children can lead to a consistent solution
(similar to~\cite{kolliopoulos2007nearly,AroraRagRao98}). Since there are at 
most $2^{O(d)}$ children parts, this gives
a running time of $q^{2^{O(d)}}$, where $q$ is the total number of parameter 
values. 

This strategy would lead to a running time of
$f(\eps,d) n \log^{2^{O(d)}} n$. We can however treat the children in 
order,
instead of naively testing all parameter values for them.
We use a classical transformation of the dynamic program, in which the first 
table
is filled using an auxiliary dynamic program. 
A cell of this auxiliary DP is a value $c_0$ in $[\cost(L)/n; 
(1+\eps)\cost(L)]$, a part $C$, one of its children $C_i$, and the portal 
parameters for the portals of $C$ and all its 
children before $C_i$ in the given order. The entry
of the cell is equal to the minimum number of centers $k_0$ that need to be 
placed in the children parts following $C_i$ to achieve a 
cost of $c_0$ given the portal parameters.
To fill this table, one can try all possible sets of parameters for the 
following children, see whether they can 
lead to a consistent solution, and compute the minimum value among them.

\paragraph{Analysis -- proof of Theorem~\ref{thm:clusteringapproxmedians} and Theorem~\ref{thm:clusteringapproxmeans}.} We first show that the solution computed by the algorithm gives a $(1+O(\eps))$-approximation, and then prove the claim on the complexity.

\begin{lemma}
Let $S^*$ be the solution computed by the algorithm. With probability $1-2\eps$, it holds that $\cost_\I(S^*) = (1+O(\eps))\cost_\I(\opt) + O(\eps \cost_\I (L))$
\end{lemma}
\begin{proof}
  With probability $1-2\eps$, $\ID$ has small distortion (Lemma~\ref{prop:structkmedian}). Following Lemma~\ref{lem:port-resp-medians}, let $S$ be a portal-respecting solution such that $\cost_{\I}(S) + B(S) \le 
  (1+O(\eps)) \cost_{\I}(\opt) + O(\eps \cost_{\I}(L))$.
  
As in Lemma~\ref{lem:alggood}, $S$ can be adapted to a configuration of the DP with a small extra cost. The cost incurred to the rounding of distances can be charged either to $B(S)$ or is a $O(\eps)\cost_{\ID}(S)$, as in Lemma~\ref{lem:alggood}. The cost to round the value $c_0$ is a $(1+\eps/\log n)$ factor at every level of the decomposition. Since there are $O(\log n)$ of them, the total factor is $(1+\eps/\log n)^{O(\log n)} = 1+O(\eps)$. Hence, we have the following:

\begin{align*}
\cost_\I (S^*) &= (1+O(\eps))\cost_{\ID}(S^*) 
&(\text{as } \ID \text{ has small distortion})\quad\\
  & = (1+O(\eps))(\cost_{\ID}(S) + B(S))  
  & (\text{by previous paragraph})\quad\\
  &  \leq (1+O(\eps))\cost_{\I}(\opt) + O(\eps \cost_{\I}(L)) 
  & (\text{by definition of } S)\quad \qedhere
  \end{align*}

\end{proof}

\begin{lemma}\label{lem:timedpmedian}
The running time of the DP is $ 2^{(1/\eps)^{O(d^2)}} \cdot n\log^4 n$.
\end{lemma}
\begin{proof}
The number of cells in the auxiliary DP is given by the number of parts ($O(n)$) , the number of children of a part ($2^{O(d)}$), the 
number of portal parameters ($(1/\eps)^{2^{O(d^2)}/\eps}$)
and the possible values for $c_0$ ($O(\log^2 n)$): it is therefore $n \cdot 2^{O(d)} \cdot (1/\eps)^{2^{O(d^2)}/\eps} \cdot \log^2 n$. 

The complexity to fill the table adds a factor $(1/\eps)^{2^{O(d^2)}/\eps} \cdot \log^2 n$, to try all possible combination of portal parameters and value of $c_0$ .
Hence, the overall running time of the DP is $n \cdot (1/\eps)^{2^{O(d^2)}/\eps} 
\cdot \log^4 n= 2^{(1/\eps)^{O(d^2)}} \cdot n\log^4 n$.
\end{proof}

The proof of Theorem~\ref{thm:clusteringapproxmedians} and Theorem~\ref{thm:clusteringapproxmeans} are completed by the following lemma, which bounds the running time of the preprocessing steps.

\begin{lemma}
For $k$-Median and $k$-Means, the total running time of the algorithms are 
respectively $2^{O(d)} n \log^9 n + 2^{(1/\eps)^{O(d^2)}}n\log^4 n$ and 
$2^{O(d)} n \log^{9} n + 2^{(1/\eps)^{O(d^2)}} n\log^5 n$
\end{lemma}
\begin{proof}
We need to bound the running time of three steps: computing an approximation, 
computing the hierarchical decomposition, and running the dynamic program.

For $k$-Median, a constant-factor approximation can be computed in $O(m \log 
^{9} n) = 2^{O(d)} n \log^9 n$ time with Thorup's 
algorithm~\cite{thorup2001quick}. The split-tree decomposition can be found in 
$2^{O(d)} n\log n$ time as explained in Section~\ref{sec:prelims}. Moreover, as 
explained in Lemma~\ref{lem:timedpmedian}, the dynamic program runs in time 
$2^{(1/\eps)^{O(d^2)}} n \log^4 n$, ending the proof of the 
Theorem~\ref{thm:clusteringapproxmedians}.

Another step is required for $k$-Means. It is indeed not known how to find a 
constant-factor approximation in near-linear time. However, one can notice that 
a $c$-approximation for $k$-Median is an $nc$-approximation for $k$-Means, using 
the Cauchy-Schwarz inequality. Moreover, notice that starting from a solution 
$S$, our algorithm finds a solution with cost $(1+O(\eps))\cost(\opt) + 
O(\eps)\cost(S)$ in time $2^{(1/\eps)^{O(d^2)}} n \log^4 n$, as for $k$-Median.

Repeating this algorithm $N$ times, using in step $i+1$ the solution given 
at step $i$, gives thus a solution of cost
$(1+O(\eps))\cost(\opt) + O(\eps^{N})\cost(S)$. Starting with $\cost(S) = 
O(n)\cost(\opt)$
and taking $N = O(\log n)$ ensures to find a solution for $k$-Means with cost
$(1+O(\eps))\cost(\opt)$. 
The complexity for $k$-Means is therefore the same as for $k$-Median, with an 
additional
$\log n$ factor for the dynamic program term. This concludes the proof of Theorem~\ref{thm:clusteringapproxmeans}.
\end{proof}

\section{Other Applications of the Framework}
\label{sec:extension}


Our techniques can be generalized to variants of the clustering problems where
\emph{outliers} are taken into account. We consider here two of them: 
$k$-Median with Outliers and its Lagrangian relaxation, Prize-Collecting 
$k$-Median. It can also be used to find a bicreteria approximation to $k$-Center.

\subsection{Prize-Collecting $k$-Median}\label{ssec:pckmed}

In the ``prize-collecting'' version of the problems,
 it is possible not to serve a client $c$ by paying a penalty $\pi_c$ (these 
problems are also called clustering ``with penalties''). For a solution $S$, we 
call an \emph{outlier for $S$} a client that is not served by $S$. Formally, an instance is a quintuple 
$(C, F, \dist, \pi, k)$ where $(C\cup F, \dist)$ is a metric, $k$ is an integer 
and 
$\pi: C \rightarrow \R^+$ the penalty function, and the goal is to find $S = (S_F, S_O)$ with
$S_F \subseteq F$ and $S_O \subseteq C$ such that $|S_F| = k$ and $\cost(S_F, S_O) = \sum_{c\in C 
\setminus S_O}\dist(c, S_F) + \sum_{c \in S_O} \pi_c$ is minimized. $\cost(S_F)$ denotes the cost of solution $S_F$ with the best choice of outliers (which is easy to determine)

Looking at the Prize-Collecting $k$-Median problem, we aim at applying the 
framework from Section~\ref{sec:structkmedian}. Let $L = (L_C, L_O)$ be an approximate solution. We define \emph{badly cut} 
for outliers as we did for centers: an outlier $c$ of $L_O$ is \emph{badly 
cut w.r.t.~$\calD$} if the ball $\beta(v, 3\opt_{c}/)$ is cut at some 
level $j$ greater than $i + \offset$, where $\opt_{c}$ is the distance from $c$ 
to the closest facility of the optimum solution $\opt$. Hence, 
Lemma~\ref{lem:badlycutddim} extends directly, and the probability that an outlier in
$L_O$ is badly cut is $\badcut$.

We now turn to the previous framework, showing how to construct a near-optimal 
solution containing all badly-cut centers of $L$. For that we transfer the 
definitions of the mappings $L_C, \psi$ ($L_C$ maps a client to its closest 
center of $L$, and $\psi(\ell) = \{f\in \opt \mid L(f) = \ell\}$)  and of the 
sets $L^0, L^1, L^{\geq 2}, \opt^1$, 
and $\opt^{\geq 2}$. We will show that this framework, with only a few modifications, 
leads to an approximation scheme for Prize-Collecting $k$-Median. Let $T = 
\opt$. As in Section~\ref{sec:structkmedian}, we start by removing a few centers from 
the optimal solution, without increasing the cost too much:

\begin{itemize}
  \item \textbf{Step 1.}
    Among the facilities of $\globalS^{\geq 2}$ that are not the closest of their 
corresponding facility in~$L^{\geq 2}$, remove from $T$ the subset 
$\widehat{\calH}$ of 
size $\lfloor \eps\cdot  |\globalS^{\geq 2}| / 2\rfloor$ that yields the smallest cost 
increase, i.e. the smallest value of $\sum_{c \in C \setminus L_O : \opt(c) \in 
\widehat{\calH}} \dist(c, T \setminus \widehat{\calH}) + \sum_{c \in  L_O : 
\opt(c) \in \widehat{\calH}} \pi_c$.
\end{itemize}

The function minimized by $\widehat{\calH}$ corresponds to redirecting all 
clients served in the local solution to a center of $T \setminus 
\widehat{\calH}$ and paying the penalty for clients $c \in L_O$ such that 
$\opt(c) \in \widehat{\calH}$. Those clients are thus considered as outliers in 
the constructed solution.

\begin{lemma}
 After step 1, $T$ has cost $(1+O(\eps))\cost(\globalS) + O(\eps)\cost(L)$
\end{lemma}
\begin{proof}[Proof sketch]
 The proof is essentially the same as Lemma~\ref{lem:coststep1}, with an 
averaging argument: the only difference comes from the cost of removing a center from $T$. For any client $c$, the cost of removing $\opt(c)$ from $T$ 
is $O(\opt_c + L_c)$: if $c \notin L^0$, the argument is the same as in 
Lemma~\ref{lem:coststep1}, and if $c \in L_O$ the cost is $\pi_c$, which is what $c$ pays in $L$.  Hence the proof follows.
\end{proof}

Again, we denote now by $\opt$ this solution $T$ and define the instance $\ID$
according to this solution. 
Recall that $B_{\calD}$ is the set of badly cut centers of $L_C$, and denote 
$O_{\calD}$ the set of badly cut outliers of $L$.
We say that an instance $\ID$ has \emph{small distortion} if
$\nu_{\ID} \leq \eps \cost(L)$ and
there exists a solution $S$ that contains
$B_{\calD}$ as centers and $O_{\calD}$ as outliers with $\cost_{\I}(S)
\le (1+\eps)\cost_{\I}(\opt) + \eps \cost_{\I}(L)$.

To deal with the badly cut centers, there is only one hurdle to be able to apply 
the proof of Lemma~\ref{lem:port-resp-medians}. Indeed, when a center of $\opt$ that serves a client $c$ is  deleted during the construction, 
the cost of reassigning $c$ is bounded in Lemma~\ref{lem:port-resp-medians} by $\dist(c, S)$. However this is not possible to 
do when $c$ is an outlier for $S$: there is no control on the cost $\dist(c, 
S)$, and hence one has to pay the penalty~$\pi_c$. It is thus necessary to find 
a mechanism that ensures to pay this penalty only with a probability $\eps$ for 
each client $c$. 
Similar to Section~\ref{sec:structkmedian}, this is achieved with the following three 
steps:
  \begin{itemize}
  \item \textbf{Step 2.} 
    For each badly-cut facility $f \in L$ for which $\psi(f)\neq\emptyset$, let 
$f' \in \psi(f)$ be the closest to~$f$. Replace $f'$ by $f$ in~$T^*$. For all 
clients $c \in L_O$ such that $\opt(c) = f'$, add $c$ as outliers. 
  \item \textbf{Step 3.}
    Add all badly cut facility $f'$ of $L^0$ to $T^*$
  \item \textbf{Step 4.} Add all badly cut outliers of $L$ to the outliers of 
$T^*$.
  \end{itemize}
We show next that $T^*$ satisfies the conditions for $\ID$ to
have small distortion with good probability.

\begin{lemma}
   The probability that $\ID$ has small distortion is at least $1-\eps$.
\end{lemma}
\begin{proof}
When bounding the cost increase due to Step 2, it is necessary to add as 
outliers all clients served by $f'$ that are outliers in $L$. Since $f'$ is 
deleted from $T^*$ with probability $\badcut$, the expected cost due to this is 
$\sum_{c \in L_O}\badcut \cdot \pi_c \leq \badcut \cost_\I(L)$. Using Markov's 
inequality, this is at most $(\eps / 3) \cost_\I(L)$ with probability 
$1-\eps/3$. 

tep 3 does not involve outliers at all. Hence,  Claim~\ref{lem:median:admiss} and 
\ref{lem:median:cost} are still valid. Combined with the previous observation 
about Step 2, this proves that after Step 3, $T^*$ contains at most $k$ centers 
--- including the ones in $B_\calD$ --- and has cost at most  
$(1+\eps)\cost_{\I}(\opt) + (\eps/3) \cost_{\I}(L)$ with probability at least 
$1-\eps / 3$.

Step 4 implies that all outliers in $O_D$ are also outliers in the constructed 
solution. Moreover, since an outlier of $L$ is badly cut with probability 
$\badcut$, the expected
cost increase due to this step is at most $\badcut \cost_\I(L)$. Using again 
Markov's inequality, this cost is at most $(\eps / 3) \cost_\I(L)$ with 
probability $1-\eps / 3$.

By union-bound, the solution $T^*$ has cost at most $(1+\eps)\cost_{\I}(\opt) + 
\eps \cost_{\I}(L)$ with probability $1-\eps$. Hence, $\ID$ has small distortion 
with probability $1-\eps$.
\end{proof}

Given an instance with low distortion, it is again possible to prove that there 
exists a near optimal portal-respecting solution, and the same DP as for 
$k$-Median can find it.

Therefore, using the polynomial time algorithm of Charikar et al.~\cite{CharikarKMN01} to compute a 
constant-factor approximation, the 
algorithm presented in Section~\ref{sec:structkmedian} can be straightforwardly adapted, 
concluding the proof of Theorem~\ref{thm:pckmed}.

\subsection{$k$-Median with Outliers}
In the $k$-Median with Outliers problem, the number of outliers allowed is 
bounded by some given integer $z$. We do not manage to respect this bound 
together with having at most $k$ facilities and a near-optimal solution: we need 
to relax it a little bit, and achieve a bicriteria approximation, with $k$ 
facilities and $(1+O(\eps))z$ outliers. For this, our framework applies nearly 
without a change. 

The first step in the previous construction does not apply directly: the 
``cost'' of removing a center is not well defined. In order to fix this part, 
Step 1 is randomized: among the facilities of $\globalS^{\geq 2}$ that are not the 
closest of their corresponding facility in~$L^{\geq 2}$, remove from $T^*$ a 
random 
subset $\widehat{\calH}$ of size $\lfloor \eps\cdot  |\globalS^{\geq 2}| / 2\rfloor$. 

\newpage
\begin{lemma}
 After the randomized Step 1, $T^*$ has expected cost 
$(1+O(\eps))\cost(\globalS) 
+ O(\eps)\cost(L)$
\end{lemma}
\begin{proof}
Since there 
are at least $|\globalS^{\geq 2}| / 2$ facilities of $\globalS^{\geq 2}$ that are not the 
closest of their corresponding facility in~$L^{\geq 2}$, the probability to remove one
of them is $O(\eps)$. Hence, every outlier of $L$ that is served in $\opt$ must 
be added as an outlier in $T^*$ with probability $O(\eps)$ -- when its serving 
center in $\opt$ is deleted. Hence, the expected number of outliers added is $O(\eps z)$.

Moreover, the proof of Lemma~\ref{lem:coststep1} shows that the sum of the cost of 
deleting all possible facilities is at most $O(\cost(\globalS)+\cost(L))$ 
(adding a point as outlier whenever it is necessary). 
Removing each one of them with probability $O(\eps)$ ensures that the expected
cost of $T^*$ after step~1 is $(1+O(\eps))\cost(\globalS) + O(\eps)\cost(L)$.
\end{proof}

The three following steps are the same as in the previous section, and the proof
follows: 
%
with constant probability, the instance $\ID$ has small distortion (defined as 
for $k$-Median with penalties), and one can use a dynamic program to solve the 
problem on it. The DP is very similar to the one for $k$-Median. The only 
difference is the addition of a number $x$ to each table entry, which is a 
power of $(1+\eps/\log(n/\eps))$, and represents the (rounded) number of 
outliers allowed in the subproblem. This adds a factor $\log ^2 (n/\eps) / 
\eps$ to the complexity.

It is possible to compute a constant factor approximation $S$ in polynomial time
(using Krishnaswamy et al.~\cite{krishnaswamy2018constant}). Hence, this algorithm is a polynomial time 
bicriteria approximation scheme for $k$-Median with outliers. As in 
Section~\ref{sec:structkmedian}, this directly extends to $k$-Means with outliers.

This concludes the proof of Theorem~\ref{thm:kmedoutliers}.

\subsection{$k$-Center}
In the $k$-Center problem, the goal is to place $k$ centers such as to minimize
the largest distance from a point to its serving center. We propose a bicriteria
approximation, allowing the algorithm to open $(1+O(\eps))k$ centers.

For this, we change slightly the definition of badly-cut. Given a solution $L$
with cost $\gamma$ and a hierarchical decomposition $\calD$,
a center $f$ of  $L$ is \emph{badly cut w.r.t $\calD$} 
if the ball $\beta(f, 2^i)$ is cut at some level $j$ greater than $i + \offset$, 
for $i$ such that $\growthrate^{i-1} \leq 2\gamma < \growthrate^i$. 

Note that Lemma~\ref{lem:badlycutddim} still holds with this definition : a center $f$
is badly cut with probability at most $\badcut$. Let $B_\calD$ be the set of badly
cut centers.
We assume in the following that $L$ is a $2$-approximation, i.e. $\gamma \leq 2\opt$.

We make the crucial following observation, using the doubling property of the metric.
Let $f$ be a center of $L$. By definition of doubling dimension, the ball $\beta(f, \gamma)$ can
be covered by $2^d$ balls of radius $\gamma / 2 \leq \opt$. Let $\calC_c$ be the 
set of centers of such balls, such that 
$\beta(f, \gamma) \subseteq \bigcup\limits_{f' \in \calC_c} \beta(f', \gamma /2)$.

Given an instance $\I$, we construct $\ID$ the following way: for each badly cut facility $f$, force
all the facilities in $\calC_f$ to be opened in any solution on $\ID$,  and remove all the 
clients in $\beta(f, \gamma)$ from
the instance. We let $\calC = \bigcup\limits_{f \text{ badly cut}} \calC_f$. 
The structural lemma of this section is the following:

\begin{lemma}\label{lem:kcenterstruc}
 It holds that for all solution $S$ of $\ID$:
 \begin{itemize}
  \item $\cost_{\ID}(S) \leq \cost_\I (S)$
  \item $\cost_\I(S \cup \calC) \leq \max(\cost_{\ID}(S), \opt)$
 \end{itemize}
\end{lemma}
\begin{proof}
  Since the instance $\ID$ contains a subset of clients of $\I$, it holds that 
  $\cost_{\ID}(S) \leq \cost_\I (S)$.
  
  Let $S$ be a solution in $\ID$. It serves all client in $\I$ but the one removed:
  these ones are served by $\calC$ at a cost $\gamma / 2 \leq \opt$. Hence, the cost 
  of $S\cup \calC$ is at most $\max(\cost_{\ID}(S), \opt)$.
\end{proof}

We now show, in a similar fashion as Lemma~\ref{lem:detour2}, that the clients in 
$\ID$ are cut from their serving facility of $\opt$ at a controlled level. Recall
that $\opt$ is defined for instance~$\I$.

\begin{lemma}\label{lem:kcenterportals}
 Let $c$ be a client in $\ID$ and $\opt(c)$ its serving facility in $\opt$.
 $C$ and $\opt(c)$ are cut at level at most $\log(2\gamma) + \offset$.
\end{lemma}
\begin{proof}
 Let $c$ be a client, $L(c)$ its serving center in $L$ and $\opt(c)$ its 
 serving center in $\opt$. If $c$ is still a client in $\ID$, it means that $L(c)$
 is not badly cut. Observe that 
 $\dist(L(c), \opt(c)) \leq \dist(c, L(c)) + \dist(c, \opt(c)) \leq \gamma + \opt \leq 2\gamma$
 
 Let $i$ such that $2^{i-1} \leq 2\gamma \leq 2^i$.
 Since $L(c)$ is not badly cut, the ball $\beta(L(c), 2^i)$ is not badly cut
 neither: hence, $c$ and $\opt(c)$ (that are in this ball) are cut at level at most
 $i + \offset \leq \log(2\gamma) + \offset$.
 \end{proof}

This lemma is stronger than Lemma~\ref{lem:detour2} and \ref{lem:detourkmedians}: it
allows us to consider only levels of the decomposition with diameter less than
$2^{1+\offset}\gamma$. 

Since the set $\calC$ has expected size $\badcut k$, Markov's inequality ensures
that with probability $1-\eps$ this set has size $O(\eps)k$. If every part with 
diameter $D$ of 
the hierarchical decomposition is equipped with a $\rho D$-net (for $\rho = \eps 
2^{-\offset}$),
Lemma~\ref{lem:kcenterportals} ensure that there exists a portal-respecting solution $S$
with cost $\cost_{\ID}(S) \leq \opt + O(\eps) \gamma = (1+O(\eps))\opt$.  
Lemma~\ref{lem:kcenterstruc} ensures that lifting this solution back to $\I$ and adding
$\calC$ as centers gives a near-optimal solution.

Using the same algorithm as for $k$-Median to compute a good portal-respecting 
solution, and computing a $2$-approximation with a simple greedy algorithm (see e.g. \cite{feder1988optimal}),
that runs in time $O(n \log k)$
concludes the proof of Theorem~\ref{thm:kcenter}.

\bibliography{papers}
\bibliographystyle{plain}

\appendix

\section{Appendix}

\subsection{Proofs for Section~\ref{sec:prelims}}\label{ap:prelims}

\begin{proof}[Proof of Lemma~\ref{lem:talwar-decomp}]
 We present the algorithm constructing the hierarchical decomposition, and 
 prove the lemma as a second step.

 Without loss of generality, assume that the smallest distance 
in the metric is $1$: the aspect-ratio $\Delta$ is therefore the diameter of 
the metric. Start from  
a hierarchy of nets $Y_0 := V, \ldots, Y_{\log(\Delta)}$ such that $Y_i$ is a 
$2^{i-2}$-net of $Y_{i-1}$. Moreover, pick a random order on the points $V$ and a 
random number $\tau \in [1/2, 1)$. The hierarchical decomposition $\calD$ is defined
inductively, starting from $\calB_{\lceil\log \Delta\rceil} = {V}$.
To partition a part $B$ at level $i$ into subparts at 
level $i-1$, do the following: for each $y \in Y_{i-1} \cap B$ taken in the random order, 
define $B \cap \beta(y, \tau 2^i)$ to be a part at level $i-1$ and remove 
$B \cap \beta(y, \tau 2^i)$ from $B$.


When we assume access to the distances through an oracle, it is possible to 
construct this hierarchy and augment it with the set of portals in time 
$(1/\rho)^{O(d)} n \log(\Delta)$. Moreover, these portals can be made \emph{nested},
meaning that portals at level $i+1$ are also portals at level $i$~\cite{har2006, 
cole2006}. 

\bigskip
We prove now that this hierarchical decomposition has the 
required properties. The diameter of each part is bounded by 
$\growthrate^{i+1}$ by construction; therefore to have Property 
\ref{prop:doub:portals} it is enough to make $\calP_i$ an $(\rho 
\growthrate^{i+1})$-net of $V$. Property~\ref{prop:doub:net} ensures the 
conciseness, and the definition of a net ensures that every point is at distance
$\rho \growthrate^{i+1}$ of $\calP_i$, which implies the preciseness. The construction ensures that a part $B$ at level $i-1$ is split in at most $|B \cap Y_{i-1}|$ parts of level $i$, which is $2^{O(d)}$ following Property~\ref{prop:doub:net}. Proving 
the 
scaling property requires a bit more work.

The two ingredients needed for this part stem from the construction of the 
decomposition: the diameter of any part at level $i$ is at most $2^{i+1}$, and the minimum 
distance between two points of $Y_i$ is bigger than $\growthrate^{i-2}$.

These two properties are enough in order to prove our lemma. Let $i$ be a level 
such that $\growthrate^i \leq r$: then $r/\growthrate^i = \Omega(1)$ so there 
is nothing to prove. Otherwise, we proceed in two steps. First, let us count 
the number of level $i$ parts that could possibly cut a ball $\beta(x, r)$. A 
level $i$ 
part is included in a ball $\beta(y, \growthrate^i)$ for some $y \in Y_i$; 
therefore if $\dist(x, y) > r + \growthrate^i$ then $y$'s part cannot cut 
$\beta(x, r)$. So it is required that $\dist(x, y) \leq r + \growthrate^i \leq 
2\cdot \growthrate^i$. But since the minimum distance between two points of $Y_i$ 
is $\growthrate^{i-2}$, and $Y_i$ has 
doubling dimension $d$, we have $|Y_i \cap \beta(x, 2\cdot \growthrate^i)| = 2^{d \log 
(\growthrate^i / \growthrate^{i-2})} = 2^{2d}$. Thus there is only a bounded 
number of parts to consider.

We prove for each of them that the probability that it cuts $\beta(x, r)$ is 
$O(r/\growthrate^i)$. A union-bound on all the possible parts is then enough 
to conclude. Let therefore $y \in Y_i \cap \beta(x, 2\cdot\growthrate^i)$, and 
$x_m$ and $x_M$ be the respective closest and farthest point of $\beta(x, r)$ 
from $y$. A necessary condition for $y$'s part to cut $\beta(x, r)$ is that 
the diameter of the part is in the open interval $(d(y, x_m), d(y, x_M))$. 
Since $x_m, x_M \in \beta(x, r)$ this interval has size $2r$, and the radius of 
the part is picked uniformly in $[\growthrate^i / 2, \growthrate^i)$. 
Therefore the probability that the radius of the part falls in $(d(y, x_m), 
d(y, x_M))$ is at most $4r / \growthrate^i$. And finally, the probability that 
$y$'s 
part cuts $\beta(x, r)$ is indeed $4r/\growthrate^i$.

By a union-bound over all the parts that could possibly cut $\beta(x,r)$ we 
obtain the claimed probability $\Pr[\calC \text{ cuts } \beta(x,r) \text{ at a 
level } i] = 2^{2d+2}r/\growthrate^i$.
\end{proof}

\begin{lemma}\label{lem:asp-rat}
 Let $P$ be a problem among Facility Location, $k$-Median or $k$-Means. 
 Given an instance $\I$ on a metric $(V, \dist)$ with $n$ points, $\eps > 0$, 
and a constant-factor approximation for $P$ on $\I$, there exists
 a linear-time algorithm that outputs a set of instances $\I_1,\ldots,\I_m$ on 
metrics $(V_1, \dist_1), \ldots, (V_m, \dist_m)$, respectively, such that
  \begin{itemize}
  \item $V_1, \ldots, V_m$ is a partition of $V$
  \item for all $i$, $(V_i,\dist_i)$ has aspect-ratio $O(n^5/\eps)$,
  \item if $(V, \dist')$ is the metric where distances 
between points of the same part $V_i$ are given by $\dist_i$ while distances 
between points of different parts is set to $\infty$, and $\opt$ is the optimum 
solution to $\I$, then 
  \begin{itemize}
	  \item there exists a solution on $(V, \dist')$ with cost
$(1+\eps/n)\cost(\opt)$, and
	  \item any solution on $(V, \dist')$ of cost $X$ induces a solution of 
cost at most $X+\eps\cost(\opt)/n$ in~$\I$.
       \end{itemize}
 \end{itemize}
\end{lemma}
\begin{proof}

The cost of the constant-factor approximation is an estimate $\gamma$ 
on the cost of the optimum solution $\opt$: $\gamma = \Theta(\cost(\opt))$. 
It is then possible to
 replace all distances longer than $2\gamma$ by $\infty$: distances longer than 
 $\gamma$ will indeed never be used by solution with cost better than $\gamma$, 
so the cost of these solutions is preserved after this transformation. The distance matrix do not respect the triangle inequality anymore: thus we replace it with its metric closure. 
We say that two vertices are \emph{connected}
if their distance is not $\infty$, and call a connected component any maximal 
set of connected vertices. The transformation ensured that any connected component
has diameter at most $2n\opt$, and that every cluster of $\opt$ is contained inside a 
single connected component. Moreover, any connected component has doubling dimension $2d$:
indeed, a subspace of a metric with doubling dimension $d$ has 
a doubling dimension at most $2d$. 
Note also that this transformation can be made implicitly: every time the algorithm
queries an edge, it can replace the result by $\infty$ if necessary. 

To identify the connected component, the algorithm builds a spanner with the algorithm
of \cite{har2006}: the connected components of the spanner are exactly the ones 
of our metric, and can be found in linear time. 

Then, for each connected component, the algorithm defines an instance of the more general version of 
the clustering problem by the following way. It first sets $\chi(v) = 1$ for all vertex $v$. 
Then, it iterates over all edges, it contracts
every edge $(u, v)$ with length less than $(\eps\cdot \gamma/n^3)$ to form a new vertex $w$
and sets $\chi(w) = \chi(u) + \chi(v)$.

Now, we aim at reconstructing a metric from this graph. We will do it in an
approximate way: for all connected points $u, v$ of connected component $i$, 
we set $\dist_i(u, v)$ to be 0 if 
$u$ and $v$ are merged in the graph, and otherwise $\dist(u, v)$. This ensures 
that $\eps\cdot \gamma / n^3 \leq \dist_i(u, v) \leq 2 n \gamma$, hence the aspect-ratio
of $\I_i$ is $O(n^5 / \eps)$. Moreover, every distance is preserved up to an
additive $O(\eps\cdot \cost(\opt) / n^2)$.

Since every cluster of $\opt$ is contained inside a single connected component,
this ensures that $\opt$ induces a solution of cost $(1+\eps / n)\cost(\opt)$
on $\bigcup \I_i$. Moreover, lifting a solution in $\bigcup \I_i$ to $\I$ costs 
at most $\eps \cost(\opt)/n^2$ per pair (client, center) and therefore $\eps 
\cost(\opt)/n$ in total.
\end{proof}

If the problem considered is Facility Location, it is easy to merge the solutions
on subinstances: since there is no cardinality constraint, the global solution
is simply the union of all the solutions. The hard constraint on $k$ makes things
a bit harder. Note that the dynamic program presented in Section~\ref{sec:structkmedian} naturally
handles it without any increase in its complexity: however, for completeness we present
now a direct reduction.

\begin{lemma}\label{lem:asp-rat2}
 Given a problem $P$ among $k$-Median or $k$-Means, a set of instances $(\I_1, 
\dist_1)$, $\ldots$, $(\I_m, \dist_m)$ given by Lemma~\ref{lem:asp-rat} and an 
algorithm running in time $n_i (\log n_i)^{\alpha} t(\Delta)$ to solve $P$ on 
instances with $n_i$ points and aspect-ratio $\Delta$, there exists an algorithm 
that runs in time $O(n (\log n)^{\alpha+2} t(O(n^4/\eps)))$ to solve $P$ on 
$\bigcup \I_i$.
\end{lemma}
\begin{proof}
 First, note that the optimal solution in $\bigcup \I_i$ is $O(n^5/\eps)$, 
since the
 maximal distance in any of $\I_1, \ldots \I_m$ is $n^4/\eps$. Using this fact, 
 we build a simple dynamic program to prove the lemma. For all $i \leq m$ and 
 $j \leq \log_{1+\eps / \log n}(n^5/\eps)$,
 let $k_{i,j}$ be the minimal $k'$ such that the cost of $P$ with $k'$ centers in 
 $\I_i$ is at most $(1+\eps / \log n)^j$. $k_{i,j}$ can be computed with a simple binary 
 search, using the fact that the cost of a solution is decreasing with $k'$.
 
 Given all the $k_{i,j}$, a simple dynamic program can compute $k_{\geq i, j}$,
 the minimal number of centers needed to have a cost at most $(1+\eps)^j$ on 
$\I_i, \ldots \I_m$
 (the $\eps/\log n$ becomes a simple $\eps$ because of the accumulation of errors).
 The solution for our problem is $(1+\eps)^j$, where $j$ is the minimal index
 such that $k_{\geq 1, j} \leq k$.
 
 The complexity of computing $k_{i,j}$ is $O(\log k \cdot n_i (\log n)^{\alpha} t(O(n^4/\eps)))$,
 hence the complexity of computing all the $k_{i,j}$ is $O(n (\log n)^{\alpha+2} t(O(n^4/\eps))$.
 The complexity of the dynamic program computing $k_{\geq i, j}$ is then simply 
 $O(m \log n) = O(n \log n)$, which concludes the proof.
\end{proof}

\subsection{Portal Respecting Paths and Solutions}\label{def:portal-resp}

Recall that any part $B\in\calB_i$ of the decomposition 
$\calD=\{\calB_0,\ldots,\calB_{|\calD|}\}$ comes with a set of portals $\calP_B$ 
with the properties listed in Lemma~\ref{lem:talwar-decomp}.
In a portal-respecting solution, every client connects to a facility by going in 
and out of parts of the decomposition only at designated portals. More 
concretely, a \emph{path} in a metric between two nodes $u$ and $v$ is given by 
a sequence of nodes $w_1, \ldots, w_k$, where $u = w_1, v = w_k$, and its length 
is $\sum_{j=1}^{k-1} \dist(w_j, w_{j+1})$. A solution can be seen as a set of 
facilities, together with a path for each client that connects it to a facility, 
and the cost of the solution is given by the sum over all path lengths. We say a 
path $w_1, \ldots, w_k$ is \emph{portal-respecting} if for every pair 
$w_j,w_{j+1}$, whenever $w_j$ and $w_{j+1}$ lie in different parts 
$B,B'\in\calB_i$ of the decomposition $\calD$ on some level $i$, then these 
nodes are also portals at this level, i.e., 
$w_j,w_{j+1}\in\calP_B\cup\calP_{B'}$. As explained in Lemma~\ref{lem:path-port-resp}, if two vertices $u$ and $v$ are cut at level $i$, then there exists a portal-respecting path from $u$ to $v$ of length at most $\dist(u,v) + 16\rho 2^i$.
We define a \emph{portal-respecting} solution to 
be a solution such that each path from a client to its closest facility in the 
solution is portal-respecting.

\begin{lemma}\label{lem:path-port-resp}
If two vertices $u$ and $v$ are cut by a decomposition at level $i$, there exists a portal-respecting path of length $\dist(u,v) + 16\rho \growthrate^i$ that connects $u$ to $v$.
\end{lemma}
\begin{proof}
If $u$ and $v$ are cut on level $i$, then they lie in the same 
part $B\in\calB_{i+1}$ on level $i+1$ of the decomposition~$\calD$. As each 
part on level $0$ of $\calD$ is a singleton set, both $u$ and $v$ are portals on 
that level. Now consider the path that starts in $u=w_1$, and for each $j\geq 1$ 
connects $w_j$ to the closest portal $w_{j+1}\in\calP_B$ of the part 
$B\in\calB_j$ on the next level $j$, until level $i+1$ is reached. This yields 
a portal-respecting path $P_u$, as portals are nested, i.e., each $w_j$ is a 
portal on every level less than $j$. A similar procedure finds a 
portal-respecting path $P_v$ from the other endpoint $v$ up to level $i+1$ 
through portals of levels below $i+1$. Since $u$ and $v$ lie in the same part on 
level $i+1$, we may obtain a portal-respecting path from $u$ to $v$ by first 
following $P_u$ up to level $i+1$, then connecting to the endpoint of $P_v$ that 
is a portals of level $i+1$, and then following $P_v$ all the way down to $v$. 
The length of this portal-respecting path is at most $\dist(u,v)+4\sum_{j\leq 
i+1}\rho 2^{j+1}=\dist(u,v)+O(\rho 2^i)$, due to the triangle inequality and 
the preciseness property of the portals (cf.~Lemma~\ref{lem:talwar-decomp}).
\end{proof}

\end{document}